%% file: IEEETAC.tex
\documentclass[12pt,draftcls,onecolumn]{IEEEtran}
\usepackage{amsmath,epsfig, cite, cases}
\usepackage{subfigure,wrapfig}
\usepackage{amssymb,amscd,mathrsfs}
\usepackage[active]{srcltx}
\usepackage{color}
\usepackage{setspace}
\usepackage{enumerate}
\usepackage{bbding}
\usepackage{algorithm, algorithmic}

\newtheorem{thm}{Theorem}[section]
\newtheorem{cor}[thm]{Corollary}
\newtheorem{lem}[thm]{Lemma}

\newtheorem{rem}{Remark}
\newtheorem{define}{Definition}

\def\pc{H}
\def\gen{G}
\def\bh{\zeta}
\def\t{t}
\def\Fbar{\overline{F}}

\def\GF{\mathbb{GF}}

\newcommand{\minW}[2]{w^{#2}_{\min,#1}}
\newcommand{\minWi}[1]{\minW{#1}{}}
\newcommand{\numW}[3]{N_{#1,#2}^{#3}}
\newcommand{\numWi}[2]{\numW{#1}{#2}{}}
\newcommand{\qedcustom}{\hfill \raggedright $\Box$}
\def\E{\mathbb{E}}
\def\mE{\mathcal{E}}

\def\Re{{\mathbb{R}}}
\def\mones{\mathbf{1}_{m_x}}

\newcommand{\argmin}{\operatornamewithlimits{argmin}}

\newcommand{\bra}[1]{\left(#1\right)}

\providecommand{\OO}[1]{\operatorname{O}\bigl(#1\bigr)}

\title{Error Correcting Codes for Distributed Control}

\author{\IEEEauthorblockN{Ravi Teja Sukhavasi and} \and \IEEEauthorblockA{Babak Hassibi}
\thanks{Ravi Teja Sukhavasi is a graduate student with the department of Electrical Engineering, California Institute of Technology, Pasadena, USA
        {\tt\small teja@caltech.edu}}%
\thanks{Babak Hassibi is a faculty with the department of Electrical Engineering, California Institute of Technology,
        Pasadena, USA
        {\tt\small hassibi@caltech.edu}}%
\thanks{This work was supported in part by the National Science Foundation under grants CCF-0729203, CNS-0932428 and CCF-1018927, by the Office of Naval Research under the MURI grant N00014-08-1-0747, and by Caltech's Lee Center for Advanced Networking.}
}
\begin{document}

\maketitle
\thispagestyle{empty}
\pagestyle{empty}

\begin{abstract}
The problem of stabilizing an unstable plant over a noisy communication link is an increasingly important one that arises in applications of networked control systems. Although the work of Schulman and Sahai over the past two decades, and their development of the notions of \lq\lq tree codes\rq\rq\phantom{} and \lq\lq anytime capacity\rq\rq, provides the theoretical framework for studying such problems, there has been scant practical progress in this area because explicit constructions of tree codes with efficient encoding and decoding did not exist. To stabilize an unstable plant driven by bounded noise over a noisy channel one needs real-time encoding and real-time decoding and a reliability which increases exponentially with decoding delay, which is what tree codes guarantee. We prove that linear tree codes occur with high probability and, for erasure channels, give an explicit construction with an expected decoding complexity that is constant per time instant. We give novel sufficient conditions on the rate and reliability required of the tree codes to stabilize vector plants and argue that they are asymptotically tight. This work takes an important step towards controlling plants over noisy channels, and we demonstrate the efficacy of the method through several examples. 
\end{abstract}

%
\input{Introduction} 				
\input{Background}   				
\input{problemSetup} 				
\input{Linear_Anytime_Codes} 			
\input{Existence_of_Linear_Anytime_Codes} 	
\input{ImprovedThresholds} 			
\input{DecodingBECCorrected} 				
\input{sufficientConditions} 			
\input{sufficientConditionsVec} 		
\input{Discussion} 				
\input{Simulations} 				
\bibliographystyle{IEEEbib}
\bibliography{IEEETAC}
\input{Appendix}

\end{document}

%% file: Introduction.tex
\section{Introduction}
\label{sec: Introduction}
Control theory deals with regulating the behavior of dynamical systems using real-time output feedback. Most traditional control systems are characterized by the measurement and control subsystems being co-located. Hence, there were no loss of measurement and control signals in the feedback loop. There is a very mature theory for this setup and there are concrete theoretical tools to analyze the overall system performance and its robustness to modeling errors \cite{Hinfinity}. There are increasingly many applications of networked control systems, however, where the measurement and control signals are communicated over noisy channels. Some examples include the smart grid, distributed computation, intelligent highways, etc (e.g., see \cite{Murray}). 

Applications of networked control systems represent different levels of decentralization in their structure. At a high level, the measurement unit and the controller are not co-located but each is individually centralized. In addition, the measurement and control subsystems are themselves comprised of arrays of sensors and actuators that in turn communicate with each other over a network. Our focus is on the former. We consider the setup where the measurement and control subsystems are individually centralized but are separated by communicated channels. 

Several aspects of this problem have been studied in the literature \cite{borkar, WongI, WongII, Nillson, Walsh}. When the communication links are modeled as rate-limited noiseless channels, significant progress has been made (see e.g.,\cite{Nair, Matveev, Minero}) in understanding the bandwidth requirements for stabilizing open loop unstable systems. \cite{Martins} considered robust feedback stabilization over communication channels that are modeled as variable rate digital links where the encoder has causal knowledge of the number of bits transmitted error free. Under a packet erasure model, \cite{Sinopoli} studied the problem of LQG (Linear Quadratic Gaussian) control in the presence of measurement erasures and showed that closed loop mean squared stability is not possible if the erasure probability is higher than a certain threshold. So, clearly the measurement and control signals need to be encoded to compensate for the channel errors. 

There are two key differences between the communication paradigm for distributed control and that traditionally studied in information theory. Shannon's information theory, in large part, is concerned with reliable one-way communication while communication for control is fundamentally interactive: the plant measurements to be encoded are determined by the control inputs, which in turn are determined by how the controller decodes the corrupted plant measurements. Furthermore, conventional channel codes achieve reliability at the expense of delay which, if present in the feedback loop of a control system, can adversely affect its performance. 

In this context, \cite{Sahai} provides a necessary and sufficient condition on the communication reliability needed over channels that are in the feedback loop of unstable scalar linear processes, and proposes the notion of anytime capacity as the appropriate figure of merit for such channels. In essence, the encoder is causal and the probability of error in decoding a source symbol that was transmitted $d$ time instants ago should decay exponentially in the decoding delay $d$. 

Although the connection between communication reliability and control is clear, very little is known about error-correcting codes that can achieve such reliabilities. Prior to the work of \cite{Sahai}, and in the context of distributed computation, \cite{Schulman} proved the existence of codes which under maximum likelihood decoding achieve such reliabilities and referred to them as tree codes. Note that any real-time error correcting code is causal and since it encodes the entire trajectory of a process, it has a natural tree structure to it. \cite{Schulman} proves the existence of nonlinear tree codes and gives no explicit constructions and/or efficient decoding algorithms. \cite{Palaiyanur} and \cite{Schulman} also propose sequential decoding algorithms whose expected complexity per time instant is fixed but the probability that the decoder complexity exceeds $C$ decays with a heavy tail as $C^{-\gamma}$. Much more recently \cite{Ostrovsky} proposed efficient error correcting codes for unstable systems where the state grows only polynomially large with time. When the state of an unstable scalar linear process is available at the encoder and when there is noiseless feedback of channel outputs, \cite{Yuksel} and \cite{Simsek} develop encoding-decoding schemes that can stabilize such a process over the binary symmetric channel and the binary erasure channel respectively. But when the state is available only through noisy measurements or when there is no channel feedback, little is known in the way of stabilizing an unstable scalar linear process over a stochastic communication channel.

The subject of error correcting codes for control is in its relative infancy, much as the subject of block coding was after Shannon's seminal work in \cite{Shannon}. So, a first step towards realizing practical encoder-decoder pairs with anytime reliabilities is to explore linear encoding schemes. We consider rate $R=\frac{k}{n}$ causal linear codes which map a sequence of $k$-dimensional binary vectors $\{b_\tau\}_{\tau=0}^\infty$ to a sequence of $n-$dimensional binary vectors $\{c_\tau\}_{\tau=0}^\infty$ where $c_t$ is only a function of $\{b_\tau\}_{\tau=0}^t$. Such a code is anytime reliable if at all times $t$ and delays $d\geq d_o$, $P\bigl(\hat{b}_{t-d|t}\neq b_{t-d}\bigr) \leq \eta 2^{-\beta nd}$ for some $\beta > 0$. We show that linear tree codes exist and further, that they exist with a high probability. For the binary erasure channel, we propose a maximum likelihood decoder whose average complexity of decoding is constant per each time iteration and for which the probability that the complexity at a given time $t$ exceeds $KC^3$ decays exponentially in $C$. This allows one to stabilize a partially observed unstable scalar linear process over a binary erasure channel and to the best of the authors' knowledge, this has not been done before. 

In Section \ref{sec: Background}, we present some background and motivate the need for anytime reliability with a simple example. In Section \ref{sec: Sufficient Condition}, we come up with a sufficient condition for anytime reliability in terms of the weight distribution of the code. In Section \ref{sec: Existence}, we introduce the ensemble of time invariant codes and use the results from Section \ref{sec: Sufficient Condition} to prove that time invariant codes with anytime reliability exist with a high probability. In Section \ref{sec: ImprovedThresholds}, we invoke some standard results from the literature on coding theory to improve the results obtained in Section \ref{sec: Existence}. In Section \ref{sec: DecodingBEC}, we present a simple decoding algorithm for the erasure channel. 


%% file: Background.tex
\section{Background}
\label{sec: Background}
Owing to the duality between estimation and control, the essential complexity of stabilizing an unstable process over a noisy communication channel can be captured by studying the open loop estimation of the same process. We will motivate the kind of communication reliability needed for control through a simple example.

\textit{A toy example: }Consider tracking the following random walk, $x_{\t+1} = \lambda x_{\t} + w_\t$, where $w_t$ is Bernoulli$\bra{\frac{1}{2}}$, i.e., 0 or 1 with equal probability, $x_0 = 0$ and $|\lambda| > 1$. Suppose an observer observes $x_\t$ and communicates over a noisy communication channel to an estimator. Also assume that the estimator knows the system model and the initial state $x_0 = 0$. The observer clearly needs to communicate whether $w_\t$ is $0$ or $1$. Note that the observer only has causal access to $\{w_i\}$, i.e., at any time $\t$, the observer has access to $\{w_0,\ldots, w_{\t-1}\}$. Let the encoding function of the observer at time $\t$ be $f_\t : \GF_2^\t\mapsto \mathcal{X}^n$, where $\mathcal{X}$ is the channel input alphabet and $n$ is the number of channel uses available for each step of the system evolution. One can visualize such a causal encoding process over a binary tree as in Fig. \ref{fig: treeCode}. While the information bits determine the path in the tree, the label on each branch denotes the symbol transmitted by the observe/encoder. The codeword associated to a given path in the tree is given by the concatenation of the branch symbols along that path.
\begin{figure}
 \centering
\includegraphics[scale=0.4]{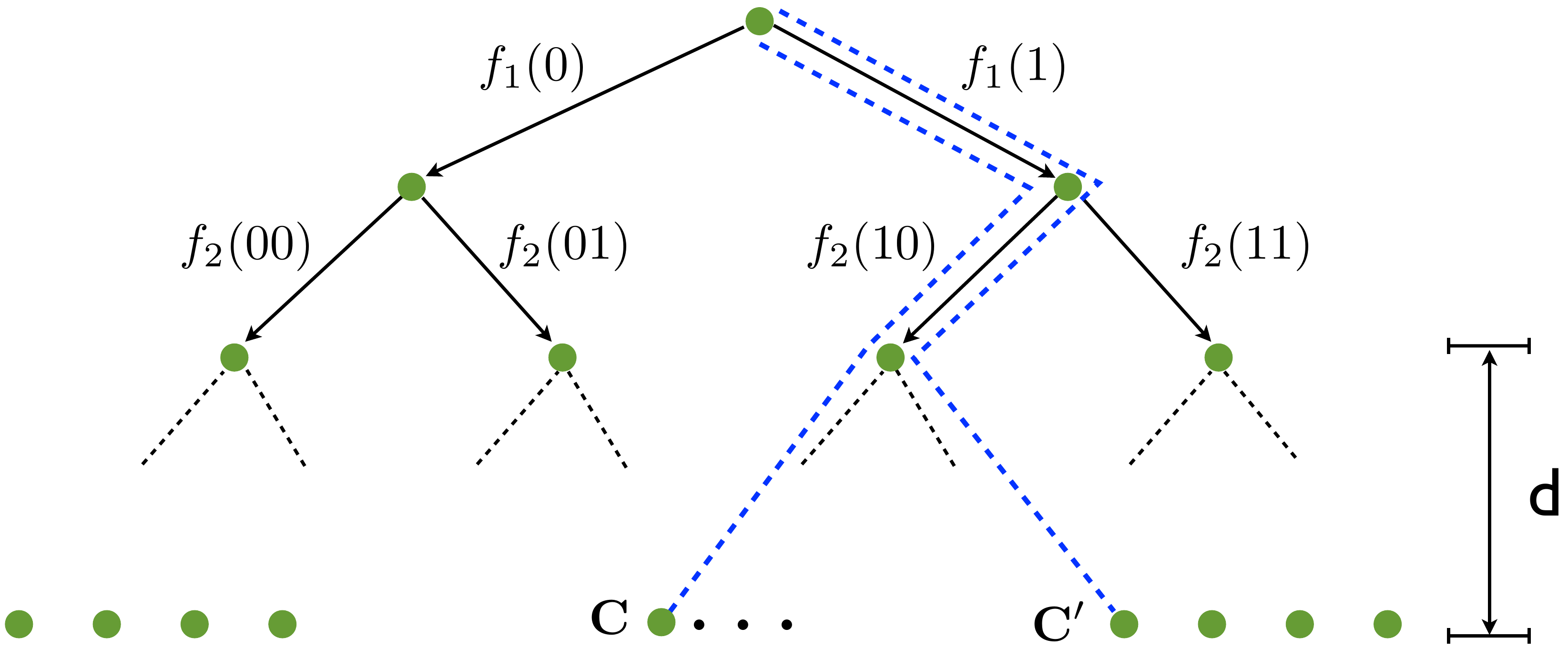}
\caption{One can visualize any causal code on a tree. The distance property is: $\|\mathbf{C}-\mathbf{C'}\|_{\mathcal{H}} \propto d$. This must be true for any two paths with a common root and of equal length in the tree}
\label{fig: treeCode}
\end{figure}
Upon receiving the channel outputs until time $t$, the estimator generates estimates $\{\hat{w}_{0|t},\hat{w}_{1|t},\ldots,\hat{w}_{t-1|t}\}$ of the noise sequence $\{w_0,w_1,\ldots,w_{t-1}\}$. Then, the estimator's estimate of the state, $\hat{x}_{\t+1|\t}$, is given by
\begin{align}
\label{eq: example}
 \hat{x}_{\t+1|\t} = \sum_{j=0}^\t\lambda_{\t-j}\hat{w}_{j|\t}
\end{align}
Suppose $P_{d,\t}^e = P\bra{\argmin_{j}(\hat{w}_{j|\t}\neq \hat{w}_{j}) = \t-d+1}$, i.e., $P_{d,t}^e$ is the probability that the position of the earliest erroneous $\hat{w}_{j|t}$ is at time $j=t-d+1$. The probability here is over the randomness of the channel. From \eqref{eq: example}, we can bound $\E\bigl|x_{\t+1}-\hat{x}_{\t+1|\t}\bigr|^2$ from above as
\begin{align*}
&\sum_{w_{0:\t},\hat{w}_{0:\t|\t}}P\bra{w_{0:\t},\hat{w}_{0:\t|\t}}\biggl| \sum_{j=1}^n\lambda^{\t-j}(w_j-\hat{w}_{j|\t})\biggr|^2\\
& \leq\sum_{d\leq \t}P_{d,\t}^e \biggl| \sum_{j=\t-d+1}^\t\lambda^{\t-j}(w_j-\hat{w}_{j|\t})\biggr|^2\\
& \leq \frac{1}{(|\lambda|-1)^2}\sum_{d\leq \t}P_{d,\t}^e|\lambda|^{2d}
\end{align*}
Clearly, a sufficient condition for $\limsup_\t\E\left|x_{\t+1}-\hat{x}_{\t+1|\t}\right|^2$ to be finite  is as follows
\begin{align}
\label{eq: anytimeReliability1} P_{d,\t}^e \leq |\lambda|^{-(2+\delta)d}\,\,\, \forall\,\,\, d\geq d_o,\,\,\, \t > \t_o\,\,\, \text{and } \delta > 0
\end{align}
where $d_o$ and $\t_o$ are constants that do no depend on $\t,d$. 

In the context of control, it was first observed in \cite{Sahai} that exponential reliability of the form \eqref{eq: anytimeReliability1} is required to stabilize unstable plants over noisy communication channels. For a given channel, encoder-decoder pairs that achieve \eqref{eq: anytimeReliability1} are said to be anytime reliable. This definition will be made more precise in Section \ref{sec: problemSetup}. In the context of distributed computation, it was observed in \cite{Schulman} that a causal code under maximum likelihood decoding over a discrete memoryless channel is anytime reliable provided that the code has a certain distance property which is illustrated in Fig. \ref{fig: treeCode}. Avoiding mathematical clutter, one can describe the distance property as follows. For any two paths with a common root and of equal length in the tree whose least common ancestor is at a height $d$ from the bottom, the Hamming distance between their codewords should be proportional to $d$. \cite{Schulman} referred to codes with this distance property as \textit{tree codes} and showed that they exist. There has recently been increased interest (e.g., \cite{Braverman, Gelles, Moitra}) in studying tree codes for interactive communication problems. But the tree codes are, in general, non-linear and the existence was not with high probability. 

We will prove the existence, with high probability, of \textit{linear tree codes} and exploit the linearity to develop an efficiently decodable anytime reliable code for the erasure channel.

%% file: problemSetup.tex
\section{Problem Setup}
\label{sec: problemSetup}

\begin{figure}
\begin{center}
\includegraphics[scale=0.35]{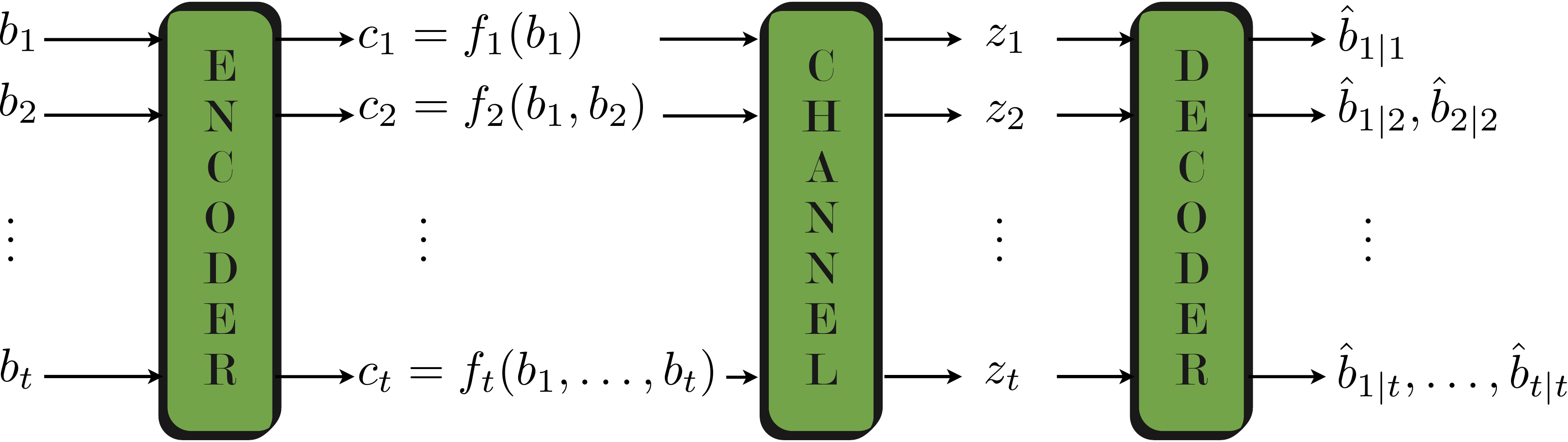}
\caption{Causal encoding and decoding}
\label{fig: operation}
\end{center}
\end{figure} 

\begin{table}
 \caption{}
 \label{tab: notation}
\begin{center}
\begin{tabular}{|r|l|}
\hline
 $H(.)$ & The binary entropy function\\
 $H^{-1}(y)$ & The smaller root of the equation $H(x) = y$\\
 For a matrix $F$, $\Fbar$ & abs($F$), i.e., $\Fbar_{i,j}=|F_{i,j}|$.$\forall$ $i,j$\\
 $\rho(F)$ & Spectral radius of $F$\\
 For a vector $x$, $x^{(i)}$ & The $i^{th}$ component of $x$\\
 $\mathbf{1}_m$ & $[1, \ldots, 1]^T$, i.e., a column with $m$ 1's\\
 For $w,v\in\Re^m$, $w \gtrless v$ & Component-wise inequality\\ 
 $\log(.)$ & Logarithm in base 2\\
 For $0\leq x,y\leq 1$, $KL\bra{x\Vert y}$ & $x\log\displaystyle\frac{x}{y} + (1-x)\log\frac{1-x}{1-y}$, i.e., Kullbeck-Leibler divergence \\
 \phantom{} & between Bernoulli($x$) and Bernoulli($y$)\\
\hline
\end{tabular}
\end{center}
\end{table}

The notation to be used in the rest of the paper is summarized in Table \ref{tab: notation}. Consider the following $m_x-$dimensional unstable linear system with $m_y-$dimensional measurements. Assume that $(F,H)$ is observable and $(F,G)$ is controllable. 
\begin{align}
\label{eq: sysmodel}
x_{t+1} = Fx_t + Gu_t + w_t,\quad y_t = Hx_t + v_t 
\end{align}
where $\rho(F)>1$, $u_t$ is the $m_u-$dimensional control input and, $w_t$ and $v_t$ are bounded process and measurement noise variables, i.e., $\Vert w_t\Vert_\infty < \frac{W}{2}$ and $\Vert v_t\Vert_\infty < \frac{V}{2}$ for all $t$. The measurements $\{y_t\}$ are made by an observer while the control inputs $\{u_t\}$ are applied by a remote controller that is connected to the observer by a noisy communication channel. We assume that the control input is available to the plant losslessly. We do not assume that the observer has access to either the channel outputs or the control inputs. As is shown to be possible, e.g., in \cite{Matveev, Sahai}, we do not use the control actions to communicate the channel outputs back to the observer through the plant because this could have a detrimental affect on the performance of the controller. 

Before proceeding further, a word is in order about the boundedness assumption on the noise. If the process and/or measurement noise have unbounded support, it is not clear how one can stabilize the system without additional assumptions on the channel. For example, \cite{Yuksel} assumes feedback of channel outputs to the observer in order to stabilize an unstable process perturbed by Gaussian noise over an erasure channel while \cite{SerdarIEEETAC2011} proposes a forward side channel between the observer and the controller that has a positive zero error capacity. We avoid this difficulty by assuming that the noise has bounded support which may be a reasonable assumption to make in practice. 

The measurements $y_{0:t-1}$ will need to be quantized and encoded by the observer to provide protection from the noisy channel while the controller will need to decode the channel outputs to estimate the state $x_t$ and apply a suitable control input $u_t$. This can be accomplished by employing a channel encoder at the observer and a decoder at the controller. For simplicity, we will assume that the channel input alphabet is binary. Suppose one time step of system evolution in \eqref{eq: sysmodel} corresponds to $n$ channel uses\footnote{In practice, the system evolution in \eqref{eq: sysmodel} is obtained by discretizing a continuous time differential equation. So, the interval of discretization could be adjusted to correspond to an integer number of channel uses, provided the channel use instances are close enough.}, i.e., $n$ bits can be transmitted for each measurement of the system. Then, at each instant of time $t$,  the operations performed by the observer, the channel encoder,  the channel decoder and the controller can be described as follows. The observer generates a $k-$bit message, $b_t\in\mathbb{GF}^k$, that is a causal function of the measurements, i.e., it depends only on $y_{0:t}$. Then the channel encoder causally encodes $b_{0:t} \in \mathbb{GF}^{kt}$ to generate the $n$ channel inputs $c_t\in\mathbb{GF}^n$. Note that the rate of the channel encoder is $R = k/n$. Denote the $n$ channel outputs corresponding to $c_t$ by $z_t \in \mathcal{Z}^n$, where $\mathcal{Z}$ denotes the channel output alphabet. Using the channel outputs received so far, i.e., $z_{0:t}\in \mathcal{Z}^{nt}$, the channel decoder generates estimates $\{\hat{b}_{\tau|t}\}_{\tau \leq t}$ of $\{b_\tau\}_{\tau\leq t}$, which, in turn, the controller uses to generate the control input $u_{t+1}$. This is illustrated in Fig. \ref{fig: operation}. Now, define
\begin{align*}
 P_{t,d}^e = P\bra{\min\{\tau:\hat{b}_{\tau|t}\neq b_\tau\} = t-d+1}
\end{align*}
Thus, $P^e_{t,d}$ is the probability that the earliest error is $d$ steps in the past.
\begin{define}[Anytime reliability]
 Given a channel, we say that an encoder-decoder pair is $(R,\beta,d_o)-$anytime reliable over that channel if
\begin{align}
\label{eq: anytimeReliability}
 P_{t,d}^e \leq \eta 2^{-n\beta d},\,\,\,\forall\,\,\,t,d\geq d_o
\end{align}
In some cases, we write that a code is $(R,\beta)-$anytime reliable. This means that there exists a fixed $d_o > 0$ such that the code is $(R,\beta,d_o)-$anytime reliable.\qedcustom
\end{define}

We will show in Sections \ref{sec: sufficient} and \ref{sec: vectorMeasurements} that $(R,\beta)-$anytime reliability with an appropriately large rate, $R$, and exponent, $\beta$, is a sufficient condition to stabilize \eqref{eq: sysmodel} in the mean squared sense\footnote{can be easily extended to any other norm}. In what follows, we will demonstrate causal linear codes which under maximum likelihood (ML) decoding achieve such exponential reliabilities.

%% file: Linear_Anytime_Codes.tex
\section{Linear Anytime Codes}
\label{sec: Sufficient Condition}
As discussed earlier, a first step towards developing practical encoding and decoding schemes for automatic control is to study the existence of linear codes with anytime reliability. We will begin by defining a causal linear code.
\begin{define}[Causal Linear Code]
 A causal linear code is a sequence of linear maps $f_\tau:\mathbb{GF}_2^{k\tau}\mapsto \mathbb{GF}_2^n$ and hence can be represented as
\begin{align}
\label{eq: GnR}
 f_\tau(b_{1:\tau}) = \gen_{\tau 1}b_1 + \gen_{\tau 2}b_2 + \ldots + \gen_{\tau \tau}b_\tau
\end{align}
where $\gen_{ij}\in \mathbb{GF}_2^{n\times k}$\qedcustom
\end{define}
We denote $c_\tau \triangleq f_\tau(b_{1:\tau})$. Note that a tree code is a more general construction where $f_\tau$ need not be linear. Also note that the associated code rate is $R = k/n$.  The above encoding is equivalent to using a semi-infinite block lower triangular generator matrix $\mathbb{G}_{n,R}$ given by
\begin{align*}
 \mathbb{G}_{n,R} = \left[\begin{array}{ccccc}
				G_{11} & 0 & \ldots & \ldots & \ldots \\
				G_{21} & G_{22} & 0 & \ldots & \ldots \\
				\vdots & \vdots & \ddots & \vdots & \vdots \\
				G_{\tau 1} & G_{\tau 2} & \ldots & G_{\tau\tau} & 0\\
				\vdots & \vdots & \vdots & \vdots & \ddots
				\end{array}\right]
\end{align*}
One can equivalently represent the code with a parity check matrix $\mathbb{H}_{n,R}$, where $\mathbb{G}_{n,R}\mathbb{H}_{n,R} = 0$. The parity check matrix is in general not unique but it is easy to see that one can choose $\mathbb{H}_{n,R}$ to be block lower triangular too.
\begin{align}
\label{eq: paritycheck}
  \mathbb{H}_{n,R} = \left[\begin{array}{ccccc}
				\pc_{11} & 0 & \ldots & \ldots & \ldots \\
				\pc_{21} & \pc_{22} & 0 & \ldots & \ldots \\
				\vdots & \vdots & \ddots & \vdots & \vdots \\
				\pc_{\tau 1} & \pc_{\tau 2} & \ldots & \pc_{\tau\tau} & 0\\
				\vdots & \vdots & \vdots & \vdots & \ddots
				\end{array}\right]
\end{align}
where $\pc_{ij}\in\{0,1\}^{\overline{n}\times n}$ and $\overline{n}=n(1-R)$. In fact, we present all our results in terms of the parity check matrix. Before proceeding further, some of the notation specific to coding is summarized in Table \ref{tab: notation2}.

\begin{table}
 \caption{}
 \label{tab: notation2}
\begin{center}
\begin{tabular}{|r|l|}
\hline
 $\mathbb{H}_{n,R}^t$ & $\overline{n}t\times nt\text{ leading principal minor of }\mathbb{H}_{n,R}$\\
 $\mathcal{C}_t$ & $\left\{c\in\{0,1\}^{nt}: \mathbb{H}_{n,R}^tc = 0\right\}$\\
 $\mathcal{C}_{t,d}$ & $\left\{c\in\mathcal{C}_t: c_{\tau<t-d+1}=0,\,\,c_{t-d+1}\neq 0\right\}$\\
 $\|c\|$ & Hamming weight of $c$\\
 $\numW{w}{d}{t}$ & $\left|\{c\in\mathcal{C}_{t,d}:\|c\|=w\}\right|$\\
 $\minW{d}{t}$ & $\argmin_{w}(\numW{w}{d}{t}\neq 0)$\\
 $P_{t,d}^e$ &  $P\bra{\min\{\tau:\hat{b}_{\tau|\t}\neq b_\tau\}=\t-d+1}$\\
\hline
\end{tabular}
\end{center}
\end{table}

The objective is to study the existence of causal linear codes which are $(R,\beta)-$anytime reliable under maximum likelihood (ML) decoding. With reference to Fig. \ref{fig: treeCode}, this amounts to choosing the branch labels, $f_\tau(b_{1:\tau})$, in such a way that they satisfy the distance property, and also are linear functions of the input, $b_{1:\tau}$. Further, we are interested in characterizing the thresholds on the rate, $R$, and exponent, $\beta$, for which such codes exist. In the interest of clarity, we will begin with a self-contained discussion of a weak sufficient condition on the distance distribution, $\{\numW{w}{d}{t},\,\minW{d}{t}\}$, of a causal linear code so that it is anytime reliable under ML decoding. This sufficient condition is an adaptation of the distance property illustrated in Fig. \ref{fig: treeCode} to the case of causal linear codes. In section \ref{sec: Existence}, we will demonstrate the existence of causal linear codes that satisfy this sufficient condition. The thresholds thus obtained will be significantly tightened in section \ref{sec: ImprovedThresholds} by invoking some standard results from random coding literature, e.g., \cite{Gallager, Barg}. 

\subsection{A Sufficient Condition}
Suppose the decoding instant is $\t$ and without loss of generality, assume that the all zero codeword is transmitted, i.e., $c_{\tau}=0$ for $\tau\leq \t$. We are interested in the error event where the earliest error in estimating $b_\tau$ happens at $\tau =\t-d+1$, i.e., $\hat{b}_{\tau|\t}=0$ for all $\tau < \t-d+1$ and $\hat{b}_{\t-d+1|\t}\neq 0$. Note that this is equivalent to the ML codeword, $\hat{c}$, satisfying $\hat{c}_{\tau < \t-d+1} = 0$ and $\hat{c}_{\t-d+1}\neq 0$, and $\mathbb{H}_{n,R}^t$ having full rank so that $\hat{c}$ can be uniquely mapped to a transmitted sequence $\hat{b}$. Then, using a union bound, we have
\begin{align}
P_{t,d}^e = P\left[\bigcup_{c\in\mathcal{C}_{t,d}}(0\text{ is decoded as }c)\right] \leq \sum_{c\in\mathcal{C}_{t,d}}P\bra{0\text{ is decoded as }c}\label{eq: weight1}
\end{align}
Consider a \textit{memoryless binary-input output-symmetric} (MBIOS) channel. Let $\mathcal{X}$ and $\mathcal{Z}$ denote the input and output alphabet respectively. The Bhattacharya parameter, $\bh$, for such a channel is defined as 
\begin{align*}
 \bh = \left\{\begin{array}{ll}
		\int\limits_{z\in\mathcal{Z}}\sqrt{p(z|X=1)p(z|X=0)}dz & \text{if }\mathcal{Z}\text{ is continuous}\\
		\sum_{z\in\mathcal{Z}}\sqrt{p(z|X=1)p(z|X=0)} & \text{if }\mathcal{Z}\text{ is discrete valued}
	     \end{array}\right.
\end{align*}
Now, it is well known (e.g., see \cite{Shamai}) that, under ML decoding 
\begin{align*}
P\bra{0\text{ is decoded as }c}\leq \bh^{\|c\|}
\end{align*}
From \eqref{eq: weight1}, it follows that $P_{t,d}^e \leq  \sum_{\minW{d}{t}\leq w\leq nd}\numW{w}{d}{t}\bh^w$.
If $\minW{d}{t} \geq \alpha nd$ and $\numW{w}{d}{t}\leq 2^{\theta w}$ for some $\theta < \log_2(1/\bh)$, then
\begin{align}
\label{eq: sufficient_condition}
 P_{t,d}^e \leq \eta2^{-\alpha nd(\log_2(1/\bh) - \theta)}
\end{align}
where $\eta = (1-2^{\log_2(1/\bh) - \theta})^{-1}$. So, an obvious sufficient condition for $\mathbb{H}_{n,R}$ can be described in terms of $\minW{d}{t}$ and $\numW{w}{d}{t}$ as follows. For some $\theta < \log_2(1/\bh)$, we need
\begin{subequations}
 \label{eq: weight_distribution}
\begin{align}
 \minW{d}{t} &\geq \alpha nd\,\,\,\forall\,\,\,t,\,\,\,d\geq d_o\\
 \numW{w}{d}{t} &\leq 2^{\theta w}\,\,\,\forall\,\,\, t,\,\,\,d\geq d_o
\end{align}
where $d_o$ is a constant that is independent of $d,t$. This brings us to the following definition
\begin{define}[Anytime distance and Anytime reliability]
 We say that a code $\mathbb{H}_{n,R}$ has $(\alpha,\theta,d_o)-$anytime distance, if the following hold
\begin{enumerate}
 \item $\mathbb{H}_{n,R}^t$ is full rank for all $t>0$
 \item $\minW{d}{t} \geq \alpha nd$, $\numW{w}{d}{t} \leq 2^{\theta w}$ for all $t > 0$ and $d \geq d_o$.\qedcustom
\end{enumerate}
\end{define}
We require that $\mathbb{H}_{n,R}^t$ have full rank so that the mapping from the source bits $b_{1:t}$ to coded bits $c_{1:t}$ is invertible. We summarize the preceeding discussion as the following Lemma.
\begin{lem}
If a code $\mathbb{H}_{n,R}$ has $(\alpha,\theta,d_o)-$anytime distance, then it is $(R,\beta,d_o)-$anytime reliable under ML decoding over a channel with Bhattacharya parameter $\bh$ where $\beta = \alpha\bra{\log(1/\bh)-\theta}$\qedcustom
\end{lem}
\end{subequations}

%% file: Existence_of_Linear_Anytime_Codes.tex
\def\pbar{\overline{p}}
\section{Linear Anytime Codes - Existence}
\label{sec: Existence}
%
%
%

Consider causal linear codes with the following Toeplitz structure
\begin{align*}
   \mathbb{H}_{n,R}^{TZ} = \left[\begin{array}{ccccc}
				\pc_{1} & 0 & \ldots & \ldots & \ldots \\
				\pc_{2} & \pc_{1} & 0 & \ldots & \ldots \\
				\vdots & \vdots & \ddots & \vdots & \vdots \\
				\pc_{\tau} & \pc_{\tau-1} & \ldots & \pc_{1} & 0\\
				\vdots & \vdots & \vdots & \vdots & \ddots
				\end{array}\right]
\end{align*}
The superscript $TZ$ in $\mathbb{H}_{n,R}^{TZ}$ denotes \lq Toeplitz'. $\mathbb{H}_{n,R}^{TZ}$ is obtained from $\mathbb{H}_{n,R}$ in \eqref{eq: paritycheck} by setting $\pc_{ij} = \pc_{i-j+1}$ for $i\geq j$. Due to the Toeplitz structure, we have the following invariance, $\minW{d}{t} = \minW{d}{t'}$ and $\numW{w}{d}{t} = \numW{w}{d}{t'}$ for all $d\leq\min(t,t')$. The code $\mathbb{H}_{n,R}^{TZ}$ will be referred to as a time-invariant code. The notion of time invariance is analogous to the convolutional structure used to show the existence of infinite tree codes in \cite{Schulman}. This time invariance allows one to prove that such codes which are anytime reliable are abundant. 

\begin{define}[The ensemble $\mathbb{TZ}_p$]
 The ensemble $\mathbb{TZ}_p$ of time-invariant codes, $\mathbb{H}_{n,R}^{TZ}$, is obtained as follows, $\pc_1$ is any fixed full rank binary matrix and for $\tau \geq 2$, the entries of $H_\tau$ are chosen i.i.d according to Bernoulli($p$), i.e., each entry is 1 with probability $p$ and 0 otherwise.
\qedcustom
\end{define}

For the ensemble $\mathbb{TZ}_p$, we have the following result
\begin{thm}[Abundance of time-invariant codes]
 \label{thm: Toeplitz}
Let $\pbar = \min\{p,1-p\}$. Then, for each $R > 0$ and
\begin{align*}
&\alpha < H^{-1}(1-R\log\bra{1/(1-\pbar)}),\,\,\,\theta > -\log\left[(1-\pbar)^{-(1-R)}-1\right],\,\,\text{we have}\\
& P\bra{\mathbb{H}_{n,R}^{TZ}\text{ has }(\alpha,\theta,d_o)-\text{anytime distance}} \geq 1-2^{-\Omega(nd_o)}
\end{align*}
\end{thm}
\begin{proof}
 See Appendix \ref{sec: proofToeplitz}
\end{proof}

We can now use this result to demonstrate an achievable region of rate-exponent pairs for a given channel, i.e., the set of rates $R$ and exponents $\beta$ such that one can guarantee $(R,\beta)$ anytime reliability using linear codes. Note that the thresholds in Theorem \ref{thm: Toeplitz} are optimal when $p=1/2$. So, for the rest of the analysis we fix $p=1/2$. To determine the values of $R$ that will satisfy \eqref{eq: sufficient_condition}, note that we need
\begin{align*}
 \log(1/(2^{1-R}-1)) < \log(1/\bh) \implies R < 1-\log(1+\bh)
\end{align*}
With this observation, we have the following Corollary.
\begin{cor}
\label{cor: thresholdsBEC}
For any rate $R$ and exponent $\beta$ such that
\begin{subequations}
\label{eq: thresholdsBEC}
\begin{align*}
 R &< 1- \log(1+\bh),\quad\text{and}\\
\beta  &< H^{-1}(1-R)\bra{\log\bra{\frac{1}{\bh}} + \log\bra{2^{1-R}-1}}
\end{align*}
\end{subequations}
if $\mathbb{H}_{n,R}^{TZ}$ is chosen from $\mathbb{TZ}_{\frac{1}{2}}$, then
\begin{align*}
 P\bra{\mathbb{H}_{n,R}^{TZ}\text{ is }(R,\beta,d_o)-\text{anytime reliable}} \geq 1-2^{-\Omega(nd_o)}
\end{align*}\qedcustom
\end{cor}
Note that for BEC($\epsilon$), $\bh = \epsilon$ and for BSC($\epsilon$), $\bh = 2\sqrt{\epsilon(1-\epsilon)}$. The constant in the exponent $\Omega(nd_o)$ in Corollary \ref{cor: thresholdsBEC} can be computed explicitly and it decreases to zero if either the rate or the exponent approach their respective thresholds. Further note that almost every code in the ensemble is $(R,\beta)$-anytime reliable after a large enough initial delay $d_o$.

The thresholds in Corollary \ref{cor: thresholdsBEC} have been obtained by using a simple union bound for bounding the error probability in \eqref{eq: weight1}. As one would expect, these thresholds can be improved by doing a more careful analysis. It turns out that the ensemble of random causal linear codes bears close resemblance to random linear block codes. This allows one to borrow results from the random coding literature to tighten the thresholds.

%% file: ImprovedThresholds.tex
\section{Improving the Thresholds}
\label{sec: ImprovedThresholds}

We will examine the Toeplitz ensemble more closely and show that its delay dependent distance distribution is bounded above by that of the random binary linear code ensemble, which we will define shortly. This will enable us to significantly improve the rate, exponent thresholds of Section \ref{sec: Existence} that were obtained using a simple union bound.

\subsection{A Brief Recap of Random Coding}
For an arbitrary discrete memoryless channel, recall the following familiar definition of the random coding exponent, $E_r(R)$, from \cite{Gallager}\footnote{We use base-2 instead of the natural logarithm}
\begin{subequations}
\begin{align}
\label{eq: randomCodingExponent}E_r(R) &= \max_{0\leq\rho\leq 1}\max_{\mathbf{Q}}\left[E_o\bra{\rho,\mathbf{Q}}-\rho R\right],\,\,\text{where}\\
\label{eq: randomCodingExponent2} E_o\bra{\rho,\mathbf{Q}} &= -\log_2\sum_{z\in\mathcal{Z}}\left[\sum_{x\in\mathcal{X}}Q(x)p(z|X=x)^{\frac{1}{1+\rho}}\right]^{1+\rho}
\end{align}
\end{subequations}

In \eqref{eq: randomCodingExponent2}, $Q(.)$ denotes a distribution on the channel input alphabet. The ensemble of random binary linear codes with block length $N$ and rate $R = \frac{K}{N}$ is obtained by choosing an $(N-K)\times N$ binary parity check matrix $H$, i.e., $H\in GF_2^{(N-K)\times N}$, each of whose entries is chosen i.i.d Bernoulli$\bra{\frac{1}{2}}$. For such an ensemble, any non-zero binary word $c\in GF_2^N$ is a codeword with probability $2^{-N(1-R)}$. For a given block code, let $w_{\min}$ denote the minimum distance and $N_w$ the number of codewords with Hamming weight $w$. A quick calculation shows that $\E N_w = \binom{N}{w}2^{-N(1-R)}$ and that $w_{\min}$ grows like $H^{-1}(1-R)N$ with a high probability. A \textit{typical} code in this ensemble is defined to be one that has $w_{\min} \approx H^{-1}(1-R)N$ and $N_w \approx \binom{N}{w}2^{-N(1-R)}$. A simple Markov inequality shows that the probability that a code from this ensemble is \textit{atypical} is at most $2^{-\Omega(N)}$. For the typical code over BSC($\epsilon$), the block error probability decays as $2^{-NE_{BSC}(R)}$ where the exponent $E_{BSC}$ has been characterized in \cite{Barg}. As has been noted in \cite{Barg}, these calculations can be easily extended to a wider class of channels. In particular, the class of MBIOS channels admits a particularly clean characterization. We present the following generalization of the result in \cite{Barg} without proof.

\begin{lem}
\label{lem: BargLemma}
Consider a linear code with block length $N$, rate $R$ and distance distribution $\{N_w\}_{w=1}^N$ such that
\begin{enumerate}
 \item $N_w = 0$ if $w \leq H^{-1}(1-R-\delta)$
 \item $N_w \leq 2^{-N(1-R-\delta+o(1))}\binom{m}{w}$
\end{enumerate}
for some $\delta > 0$. Let the channel be a MBIOS channel with Bhattacharya parameter $\bh$. Then the block error probability, $P_e$, under ML decoding is bounded as
\begin{align}
 P_e \leq 2^{-N\bra{E_\bh(R) - \delta'}}
\end{align}
where 
\begin{align}
 E_\bh(R) = \left\{\begin{array}{cc}
                    H^{-1}(1-R)\log\frac{1}{\bh} & ,\,\,\,0\leq R\leq 1 - H\bra{\frac{\bh}{1+\bh}}\\
		    E_r(R)            & ,\,\,\, 1 - H\bra{\frac{\bh}{1+\bh}}\leq R\leq C
                   \end{array}
		\right.
\end{align}
and $\delta'\rightarrow 0$ as $\delta\rightarrow 0$.
\end{lem}
\begin{IEEEproof}
 The proof is a straightforward generalization of the result in \cite{Barg}.
\end{IEEEproof}

\subsection{The Toeplitz Ensemble}
In the causal case, fix an arbitrary decoding instant $t$ and consider the event that the earliest error happens at a delay $d$. As seen before, the associated error probability depends on the relevant codebook $\mathcal{C}_{t,d}$ and its distance distribution $\{\numW{w}{d}{t}\}_{w=1}^{nd}$. Recall from Table \ref{tab: notation2} that
\begin{align*} 
\mathcal{C}_{t,d} \triangleq \left\{c\in\mathcal{C}_t: c_{\tau<t-d+1}=0,\,\,c_{t-d+1}\neq 0\right\}
\end{align*}
Due to the Toeplitz structure, we have $\mathcal{C}_{t,d} = \mathcal{C}_{d,d}$. So, we drop the subscript $t$ in $\numW{w}{d}{t}$ and write it as $\numWi{w}{d}$. Note that $\mathcal{C}_{d,d}$ is determined by the matrix $\mathbb{H}_{n,R}^d$. Let $c$ be a given $nd$-dimensional binary word, i.e., $c\in GF_2^{nd}$, and write $c = \left[c_1^T,c_2^T,\ldots,c_d^T\right]^T$, where $c_\tau\in GF_2^n$ notionally corresponds to the $n$ encoder output bits during the $\tau^{th}$ time slot. Suppose $c_1 \neq 0$, then it is easy to see that
\begin{align*}
 P\bra{\mathbb{H}_{n,R}^dc = 0} = 2^{-\overline{n}d}
\end{align*}
Recall that $\overline{n} = n(1-R)$.

Now observe that $\E N_{w,d} \leq \binom{nd}{w}2^{-\overline{n}d}$. This is same as the average weight distribution of the random binary linear code with a block length $nd$ and rate $R$. So, applying Lemma \ref{lem: BargLemma}, we get the following result.

\begin{thm}
\label{thm: improvedOverRCE}
 For each rate $R < C$ and exponent $\beta < E_\bh(R)$, if $\mathbb{H}_{n,R}^{TZ}$ is chosen from $\mathbb{TZ}_{\frac{1}{2}}$, then
\begin{align*}
 P\bra{\mathbb{H}_{n,R}^{TZ}\text{ is }(R,\beta,d_o)-\text{anytime reliable}} \geq 1-2^{-\Omega(nd_o)}
\end{align*}
where $C$ is the Shannon capacity of the channel and 
\begin{align}
 E_\bh(R) = \left\{\begin{array}{cc}
                    H^{-1}(1-R)\log\frac{1}{\bh} & ,\,\,\,0\leq R\leq 1 - H\bra{\frac{\bh}{1+\bh}}\\
		    E_r(R)            & ,\,\,\, 1 - H\bra{\frac{\bh}{1+\bh}}\leq R\leq C
                   \end{array}
		\right.
\end{align}\qedcustom
\end{thm}

The problem of stabilizing unstable scalar linear systems over noisy channels in the absence of feedback has been considered in \cite{Sahai}. \cite{Sahai} showed the existence of $(R,\beta)-$anytime reliable codes for $R < C$ and $\beta < E_r(R)$. The code is not linear in general and the existence was not with high probability. Theorem \ref{thm: improvedOverRCE} proves linear anytime reliable codes for exponent, $\beta$, up to $E_\bh(R)$. When $R < 1 - H\bra{\frac{\bh}{1+\bh}}$, $E_\bh(R) > E_r(R)$. So, Theorem \ref{thm: improvedOverRCE} marks a significant improvement in the known thresholds for stabilizing unstable processes over noisy channels, as is demonstrated in Figures \ref{fig: compareExps} and \ref{fig: BECvBSC}. 

\begin{figure}
\centering
\subfigure[Binary Erasure Channel, $\epsilon=0.15$]{\includegraphics[scale=0.19]{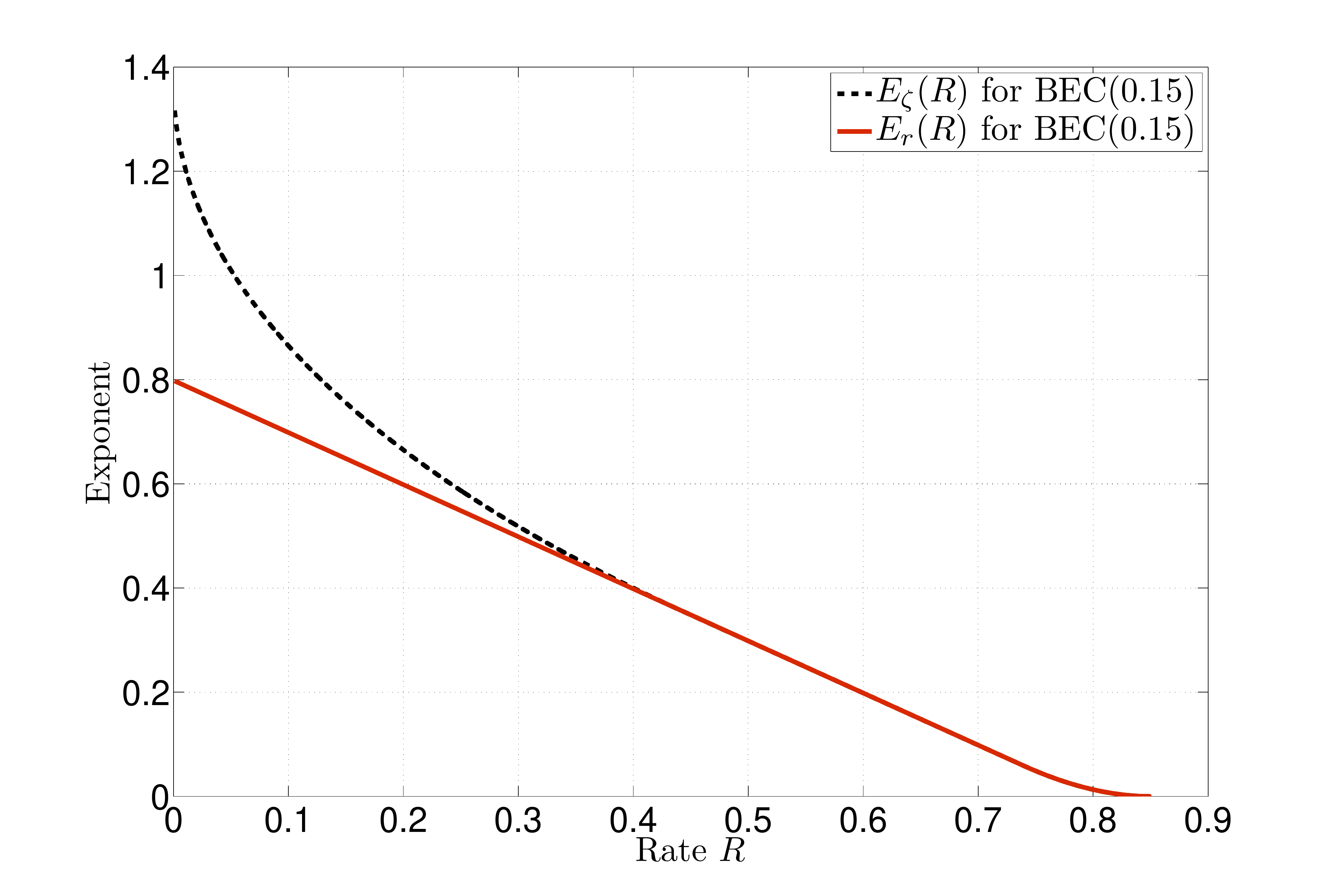}}\,\,
 \subfigure[Binary Symmetric Channel, $\epsilon=0.05$]{\includegraphics[scale=0.19]{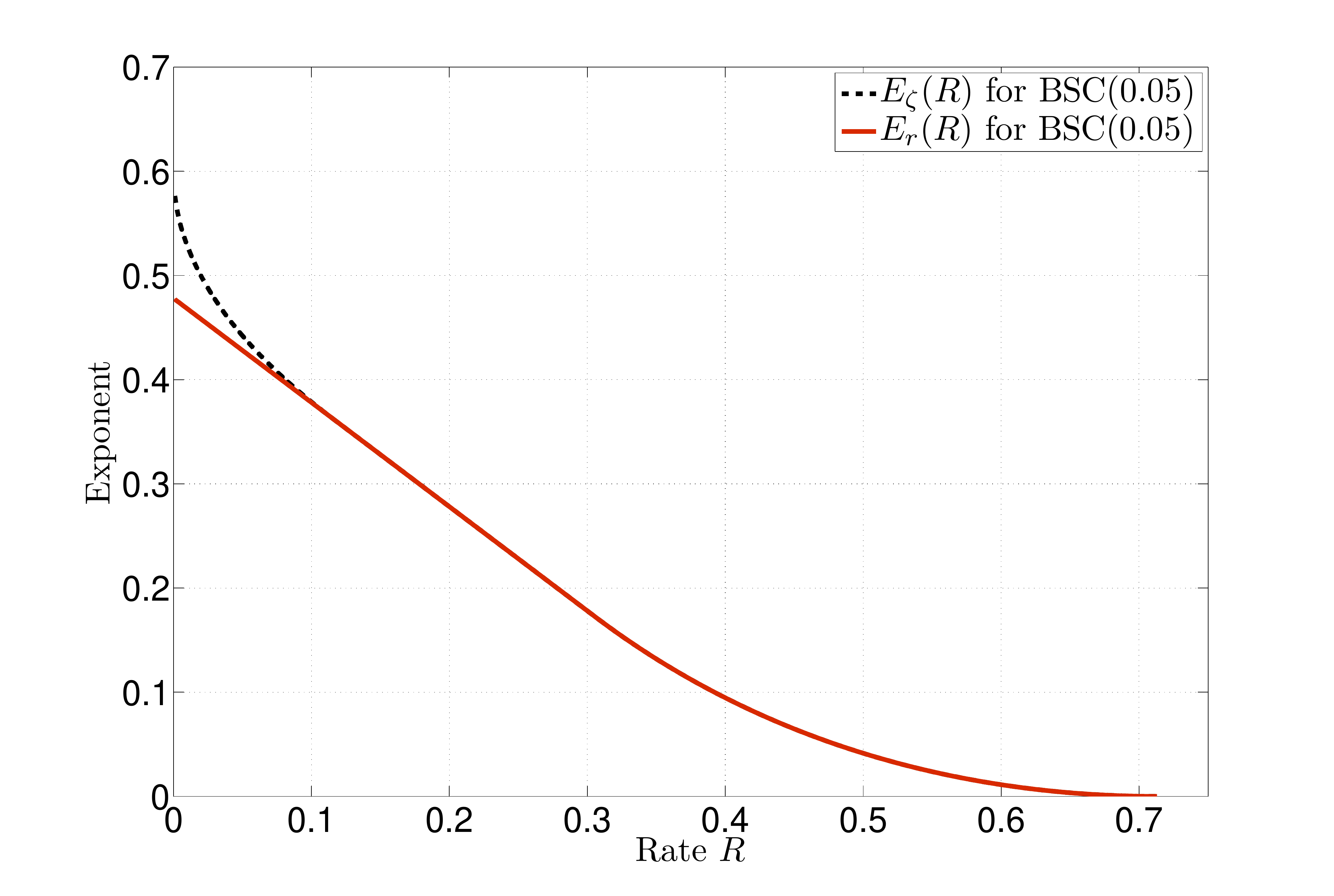}}
\caption{Comparing the thresholds obtained from Theorem \ref{thm: improvedOverRCE} and Theorem 5.2 in \cite{Sahai}}
\label{fig: compareExps}
\end{figure}

%% file: DecodingBECCorrected.tex
\section{Decoding over the Binary Erasure Channel}
\label{sec: DecodingBEC}

Owing to the simplicity of the erasure channel, it is possible to come up with an efficient way to perform maximum likelihood decoding at each time step. Consider an arbitrary decoding instant $t$, let $c=[c_1^T,\ldots,c_t^T]^T$ be the transmitted codeword and let $z=[z_1^T,\ldots,z_t^T]^T$ denote the corresponding channel outputs. Recall that $\mathbb{H}_{n,R}^t$ denotes the $\overline{n}t\times nt$ leading principal minor of $\mathbb{H}_{n,R}$. Let $z_e$ denote the erasures in $z$ and let $H_e$ denote the columns of $\mathbb{H}_{n,R}^t$  that correspond to the positions of the erasures. Also, let $\tilde{z}_e$ denote the unerased entries of $z$ and let $\tilde{H}_e$ denote the columns of $\mathbb{H}_{n,R}^t$ excluding $H_e$. So, we have the following parity check condition on $z_e$, $H_ez_e = \tilde{H}_e\tilde{z}_e$. Since $\tilde{z}_e$ is known at the decoder, $s\triangleq \tilde{H}_e\tilde{z}_e$ is known. Maximum likelihood decoding boils down to solving the linear equation $H_ez_e = s$. Due to the lower triangular nature of $H_e$, unlike in the case of traditional block coding, this equation will typically not have a unique solution, since $H_e$ will typically not have full column rank. This is alright as we are not interested in decoding the entire $z_e$ correctly, we only care about decoding the earlier entries accurately. If $z_e = [z_{e,1}^T,\,\, z_{e,2}^T]^T$, then $z_{e,1}$ corresponds to the earlier time instants while $z_{e,2}$ corresponds to the latter time instants. The desired reliability requires one to recover $z_{e,1}$ with an exponentially smaller error probability than $z_{e,2}$. Since $H_e$ is lower triangular, we can write $H_ez_e = s$ as
\begin{align}
\label{eq: bec1}
 \left[\begin{array}{cc}
  H_{e,11} & 0\\
  H_{e,21} & H_{e,22}
 \end{array}\right]\left[\begin{array}{c}z_{e,1}\\z_{e,2}\end{array}\right] = \left[\begin{array}{c}s_1\\s_2\end{array}\right]
\end{align}
Let $H_{e,22}^\bot$ denote the orthogonal complement of $H_{e,22}$, ie., $H_{e,22}^\bot H_{e,22} = 0$. Then multiplying both sides of \eqref{eq: bec1} with diag$(I,H_{e,22})$, we get
\begin{align}
 \label{eq: bec2}
 \left[\begin{array}{c}
  H_{e,11}\\
  H_{e,22}^\bot H_{e,21}
 \end{array}\right]z_{e,1} = \left[\begin{array}{c}s_1\\H_{e,22}^\bot s_2\end{array}\right]
\end{align}
If $[H_{e,11}^T\,\,\, (H_{e,22}^\bot H_{e,21})^T]^T$ has full column rank, then $z_{e,1}$ can be recovered exactly. The decoding algorithm now suggests itself, i.e., find the smallest possible $H_{e,22}$ such that $[H_{e,11}^T\,\,\, (H_{e,22}^\bot H_{e,21})^T]^T$ has full rank and it is outlined in Algorithm \ref{alg: algorithm}.
\begin{algorithm}
\caption{Decoder for the BEC}
\label{alg: algorithm}
\begin{enumerate}
\item Suppose, at time $t$, the earliest uncorrected error is at a delay $d$. Identify $z_e$ and $H_e$ as defined above.
\item Starting with $d'=1,2,\ldots,d$, partition
\begin{align*}
 z_e = [z_{e,1}^T\,\,z_{e,2}^T]^T\,\,\text{and}\,\,H_e = \left[\begin{array}{cc}H_{e,11}&0\\H_{e,21}&H_{e,22}\end{array}\right]
\end{align*}
where $z_{e,2}$ correspond to the erased positions up to delay $d'$.
\item Check whether the matrix $\left[\begin{array}{c}
  H_{e,11}\\
  H_{e,22}^\bot H_{e,21}
 \end{array}\right]$
has full column rank.
\item If so, solve for $z_{e,1}$ in the system of equations
\begin{align*}
  \left[\begin{array}{c}
  H_{e,11}\\
  H_{e,22}^\bot H_{e,21}
 \end{array}\right]z_{e,1} = \left[\begin{array}{c}s_1\\H_{e,22}^\bot s_2\end{array}\right]
\end{align*}
\item Increment $t=t+1$ and continue.
\end{enumerate}
\end{algorithm} 
Note that one can equivalently describe the decoding algorithm in terms of the generator matrix and it will be very similar to Alg \ref{alg: algorithm}.

\subsection{Encoding and Decoding Complexity}
Consider the decoding instant $t$ and suppose that the earliest uncorrected erasure is at time $t-d+1$. Then steps 2) and 3) in Algorithm \ref{alg: algorithm} can be accomplished by just reducing $H_e$ into the appropriate row echelon form, which has complexity $\OO{d^3}$. The earliest entry in $z_e$ is at time $t-d+1$ implies that it was not corrected at time $t-1$, the probability of which is $P_{d-1,t-1}^e \leq \eta 2^{-n\beta (d-1)}$. Hence, if nothing more had to be done, the average decoding complexity would have been at most $K\sum_{d>0}d^3 2^{-n\beta d}$ which is bounded and is independent of $t$. In particular, the probability of the decoding complexity being $Kd^3$ would have been at most $\eta 2^{-n\beta d}$. But, inorder to actually solve for $z_{e,1}$ in step 4), one needs to compute the syndromes $s_1$ and $s_2$. It is easy to see that the complexity of this operation increases linearly in time $t$. This is to be expected since the code has infinite memory. A similar computational complexity also plagues the encoder, for, the encoding operation at time $t$ is described by $c_t = G_tb_1+\ldots+G_1b_t$ where $\{b_i\}$ denote the source bits and hence becomes progressively hard with $t$.

We propose the following scheme to circumvent this problem in practice. We allow the decoder to periodically, say at $t = \ell (2T)$ ($\ell=1,2\ldots$) for appropriately chosen $T$, provide feedback to the encoder on the position of the earliest uncorrected erasure which is, say at time $t-d$. The encoder can use this information to stop encoding the source bits received prior to $t-d$, i.e., $\{b_i\}$ for $i \leq t-d-1$ starting from time $t+T$. In other words, for $\tau > t+T$, $c_\tau = G_{\tau-t+d+2}b_{t-d-1} + \ldots + G_1b_\tau$. The decoder accordingly uses the new generator matrix starting from $t+T$. In practice, this translates to an arrangement where the decoder sends feedback at time $t$ and can be sure that the encoder receives it by time $t+T$. Such feedback, in the form of acknowledgements from the receiver to the transmitter, is common to most packet-based modern communication and networked systems for reasonable values of $T$. Note that this form of feedback finds a middle ground between one extreme of having no feedback at all and another extreme where every channel output is fed back to the transmitter, the latter being impractical in most cases. The decoder proposed in Alg. \ref{alg: algorithm} is easy to implement and its performance is simulated in Section \ref{sec: Simulations}. 

\subsection{Extension to Packet Erasures}
The encoding and decoding algorithms presented so far have been developed for the case of bit erasures. But it is not difficult to see that the techniques generalize to the case of packet erasures. For example, for a packet length $L$, what was one bit earlier will now be a block of $L$ bits. Each binary entry in the encoding/parity check matrix will now be an $L\times L$ binary matrix. The rate will remain the same. So, at each time, $k$ packets each of length $L$  will be encoded to $n$ packets each of the same length $L$. Recall that the \textit{anytime performance} of the code is determined by the delay dependent codebook $\mathcal{C}_{t,d}$ and its distance distribution $\{\numW{w}{d}{t}\}_{w=1}^{nd}$. In the case of packet erasures, one can obtain analogous results by defining the Hamming distance of a codeword slightly differently. By viewing a codeword as a collection of packets, define its Hamming distance to be the number of non zero packets. The definition of the delay dependent distance distribution $\{\numW{w}{d}{t}\}$ will change accordingly. With this modification, one can easily apply the results developed in Sections \ref{sec: Sufficient Condition}, \ref{sec: Existence} and the decoding algrithm in Section \ref{sec: DecodingBEC} above to the case of packet erasures.

%% file: sufficientConditions.tex
\section{Sufficient Conditions for Stabilizability - Scalar Measurements}
\label{sec: sufficient}

Recall that we do not assume any feedback about the channel outputs or the control inputs at the observer/encoder. This is the setup we imply whenever we say that no feedback is assumed. In this context \cite{Sahai} derives a sufficient condition for stabilizing scalar linear systems over noisy channels without feedback while \cite{SahaiVec} considers stabilizing vector valued processes in the presence of feedback. So, to the best of our knowledge, there are no results on stabilizing unstable vector valued processes over a noisy channel when the observer does not have access to either the control inputs or the channel outputs.

We will develop two sufficient conditions for stabilizing vector valued processes over noisy channels without feedback. The two sufficient conditions are based on two different estimation algorithms employed by the controller and neither is stronger than the other. We will then show in Section \ref{subsec: theLimitingCase} that both sufficient conditions are asymptotically tight. For ease of presentation, we will treat the case of scalar and vector measurements separately. We will present the sufficient conditions for the case of scalar measurements here while vector measurements will be treated in Section \ref{sec: vectorMeasurements}

Consider the unstable $m_x-$dimensional linear state space model in \eqref{eq: sysmodel} with scalar measurements, i.e., $\rho(F) > 1$, and $m_y = 1$. Suppose that the characteristic polynomial of $F$ is given by 
\begin{align*}
 f(z) \triangleq z^{m_x} + a_1z^{m_x-1}+\ldots+a_{m_x}
\end{align*}
Without loss of generality we assume that $(F,H)$ are in the following canonical form. 
\begin{align*}
 F = \left[\begin{array}{ccccc}
		-a_1 & 1 & 0 & \ldots & \\
		-a_2 & 0 & 1 & 0 &\\
		\vdots & \vdots & &\ddots & \\
		-a_{m-1} & \ldots & \ldots&0 & 1\\
		-a_m & 0 & \ldots & \ldots & 0
		\end{array}\right],\quad H = [1,0,\ldots,0]
\end{align*}

Owing to the duality between estimation and control, we can focus on the problem of tracking \eqref{eq: sysmodel} over a noisy communication channel. For, if \eqref{eq: sysmodel} can be tracked with an asymptotically finite mean squared error and if $(F,G)$ is stabilizable, then it is a simple exercise to see that there exists a control law $\{u_t\}$ that will stabilize the plant in the mean squared sense, i.e., $\limsup_t\E\Vert x_t\Vert^2 < \infty$. In particular, if the control gain $K$ is chosen such that $F + GK$ is stable, then $u_t = K\hat{x}_{t|t}$ will stabilize the plant, where $\hat{x}_{t|t}$ is the estimate of $x_t$ using channel outputs up to time $t$. In control parlance, this amounts to verifying that the control input does not have a \textit{dual effect} \cite{Shalom}. Hence, in the rest of the analysis, we will focus on tracking \eqref{eq: sysmodel}. The control input $u_t$ therefore is assumed to be absent, i.e., $u_t=0$.
\input{Cuboidal}
\input{Ellipsoidal}

%% file: Cuboidal.tex
\subsection{Hypercuboidal Filter}
\label{sec: Cuboidal}
We bound the set of all possible states that are consistent with the estimates of the quantized measurements using a hypercuboid, i.e., a region of the form $\left\{x\in\Re^{m_x} | \mathbf{a}\leq x\leq \mathbf{b}\right\}$, where $\mathbf{a},\mathbf{b}\in\Re^{m_x}$ and the inequalities are component-wise. 

Since we assume that the initial state $x_0$ has bounded support, we can write $x_{min,0|-1}\leq x_0\leq x_{max,0|-1}$ and suppose using the channel ouputs received till time $t-1$, we have $x_{min,t|t-1}\leq x_t\leq x_{max,t|t-1}$. Since $H = [1, 0,\ldots, 0]$, the measurement update provides information of the form $x_{min,t|t}^{(1)}\leq x_t^{(1)}\leq x_{max,t|t}^{(1)}$ while there will be no additional information on other components of $x_t$. Note that an estimate of the state is given by the mid point of this region, i.e., $\hat{x}_{t|t} = 0.5(x_{min,t|t} + x_{max,t|t})$. If we define $\Delta_{t|t} = x_{max,t|t}-x_{min,t|t}$, then the estimation error is asymptotically bounded if every component of $\Delta_{t|t}$ is asymptotically bounded.
Using such a filter, we can stabilize the system in the mean squared sense over a noisy channel provided that the rate $R$ and exponent $\beta$ of the $(R,\beta)-$anytime reliable code used to encode the measurements satisfy the following sufficient condition

\begin{thm}
\label{thm: cuboidalThm}
 It is possible to stabilize \eqref{eq: sysmodel} in the mean squared sense with an $(R,\beta)-$anytime code provided 
\begin{align}
\label{eq: cuboidalThm}
 R > R_{n} = \frac{1}{n}\log_2\sum_{i=1}^{m_x} |a_i|,\,\,\,\,\beta > \beta_{n} = \frac{2}{n}\log_2\rho(\Fbar)
\end{align}
\end{thm}
\begin{proof}
 See Appendix \ref{sec: proofCuboidal}
\end{proof}

Before proceeding further, we will provide a brief sketch of the proof. Note that $\Delta_{t|t} = x_{max,t|t}-x_{min,t|t}$ is a measure of the uncertainty in the state estimate. From Lemma \ref{lem: cubTimeUpdate}, $\Delta_{t+1|t} = \Fbar\Delta_{t|t} + W\mones$. The anytime exponent is determined by the growth of $\Delta_t$ in the absence of measurements, hence the bound $\beta_{n} = 2\log_2\rho(\Fbar)$. The bound on the rate is determined by how fine the quantization needs to be for $\Delta_t$ to be bounded asymptotically. It will be shown in Section \ref{subsec: limitingCase} that $\rho\bra{\Fbar}$ is always larger than $\rho\bra{F}$. By using an alternate filtering algorithm, which we call the Ellipsoidal filter, one can improve this requirement on the exponent from $\beta_n > 2\log_2\rho(\Fbar)$ to $\beta_n > 2\log_2\rho(F)$. But this will come at the price of a larger rate.

%% file: Ellipsoidal.tex
\subsection{Ellipsoidal Filter}
\label{sec: Ellipsoidal}
One can alternately bound the set of all possible states that are consistent with the estimates of the quantized measurements using an ellipsoid 
\begin{align*}
\mE(P,c)\triangleq \left\{x\in\Re^{m_x}| \langle x-c,P^{-1}(x-c)\rangle \leq 1\right\} 
\end{align*}
This can be seen as an extension of the technique proposed in \cite{Schweppe} to filtering using quantized measurements. If $m_x=1$, $\rho(\Fbar)=\rho(F)$. So, let $m_x\geq 2$. 

Let $x_0\in\mE(P_0,0)$ and suppose using the channel outputs received till time $t-1$, we have $x_t\in\mE(P_{t|t-1},\hat{x}_{t|t-1})$.  Since $H = [1, 0,\ldots, 0]$, the measurement update provides information of the form $x_{min,t|t}^{(1)}\leq x_t^{(1)}\leq x_{max,t|t}^{(1)}$, which one may call a slab. $\mE(P_{t|t},\hat{x}_{t|t})$ would then be an ellipsoid that contains the intersection of the above slab with $\mE(P_{t|t-1},\hat{x}_{t|t-1})$, in particular one can set it to be the minimum volume ellipsoid covering this intersection. Lemma \ref{lem: minVol} gives a formula for the minimum volume ellipsoid covering the intersection of an ellipsoid and a slab. For the time update, it is easy to see that for any $\epsilon' > 0$ and $P_{t+1} = (1+\epsilon')FP_{t|t}F^T + \frac{W^2}{4\epsilon'}\mones\mones^T$, $\mE(P_{t+1},F\hat{x}_{t|t})$ contains the state $x_{t+1}$ whenever $\mE(P_{t|t},\hat{x}_{t|t})$ contains $x_t$. This leads to the following Lemma, the proof of which is contained in the discussion above. For convenience, we write $P_{t}$ for $P_{t|t-1}$.
\begin{lem}[The Ellipsoidal Filter]
Whenever $\mE(P_0,0)$ contains $x_0$, for each $\epsilon' > 0$, the following filtering equations give a sequence of ellipsoids $\left\{\mE(P_{t|t},\hat{x}_{t|t})\right\}$ that, at each time $t$, contain $x_t$. 
\begin{subequations}
 \begin{align}
 \label{eq: ellipTU} P_{t+1} &= (1+\epsilon')FP_{t|t}F^T + \frac{W^2}{4\epsilon'}\mones,\,\,\hat{x}_{t+1} = F\hat{x}_{t|t}\\
 \label{eq: ellipMU} P_{t|t} &= b_tP_t - (b_t-a_t)\frac{P_te_1e_1^TP_t}{e_1^TP_te_1},\,\,\hat{x}_{t|t} = \xi_t\frac{P_te_1}{\sqrt{e_1^TP_te_1}}
 \end{align}
\label{eq: ellip}
\end{subequations}
where $a_t,b_t$ and $\xi_t$ can be calculated in closed form using Lemma \ref{lem: minVol}, and $e_1$ is the $m_x-$dimensional unit vector $e_1 = \left[1,0,\ldots,0\right]^T$.
\end{lem}

Using this approach, we get the following sufficient condition. 
\begin{thm}
\label{thm: ellipsoidalThm}
 It is possible to stabilize \eqref{eq: sysmodel} for $m_x\geq 2$ in the mean squared sense with an $(R,\beta)-$anytime code provided
\begin{subequations}
\label{eq: ellipsoidalThm}
\begin{align}
 R &> R_{e,n} = \frac{1}{n}\log_2\left[\sqrt{m_x}\sum_{i=1}^{m_x} |a_i|\theta^{i-1}\right]\\
 \beta &> \beta_{e,n} = \frac{2}{n}\log_2\rho(F)
\end{align}
\end{subequations}
 where $\theta = \sqrt{\frac{m_x}{m_x-1}}$
\end{thm}
\begin{proof}
 See Appendix \ref{sec: proofEllipsoidal}
\end{proof}

%% file: sufficientConditionsVec.tex
\section{Sufficient Conditions for Stabilizability - Vector Measurements}
\label{sec: vectorMeasurements}

Like in the scalar case, we will assume without loss of generality that $(F,H)$ are in a canonical form (is obtained from a simple transformation of \textit{Scheme I} in Sec 6.4.6 of \cite{Kailath}) with the following structure. $F$ is a $q\times q$ block lower triangular matrix with $F^{i,j}$ denoting the $(i,j)^{th}$ block. So, $F^{i,j} = 0$ if $j > i$. $F^{i,j}$ is an $\ell_i\times\ell_j$ matrix and $\sum_{i=1}^q\ell_i = m_x$. The diagonal blocks $F^{i,i}$ have the following structure.
\begin{align*}
 F^{i,i} &= \left[\begin{array}{ccccc}
		-a_{i,1} & 1 & 0 & \ldots & \\
		-a_{i,2} & 0 & 1 & 0 &\\
		\vdots & \vdots & &\ddots & \\
		-a_{i,\ell_i-1} & \ldots & \ldots&0 & 1\\
		-a_{i,\ell_i} & 0 & \ldots & \ldots & 0
		\end{array}\right]
\end{align*}
while the off-diagonal blocks do not have any specific structure. The measurement matrix $H$ is of the form $H = \left[H_{1}^T,\,\, H_{2}^T\right]^T$ where $H_{1}$ is a $q\times m_x$ matrix of the following form
\begin{align}
 H_{1} = \text{block diag}\left\{\left[1\,\,0\,\,\cdots\,\,0\right],\,\,1\times\ell_i,\,\,i=1,\ldots,q\right\}
\end{align}
$H_{2}$ does not have any particular structure and is not relevant. Note that the characteristic polynomial of $F$, is given by $f(z) = \prod_{i=1}^q\bra{z^{\ell_i}+a_{i,1}z^{\ell_i-1}+\ldots+a_{i,\ell_i}}$.

If the Hypercuboidal filter is used, then Theorem \ref{thm: cuboidalThm} can be extended to the case of vector measurements is as follows.
\begin{thm}
\label{thm: cuboidalThmVec}
 It is possible to stabilize \eqref{eq: sysmodel} in the mean squared sense with an $(R,\beta)-$anytime code provided 
\begin{subequations}
\begin{align}
\label{eq: cuboidalThmVec}
 R > R_{v,n} &= \frac{1}{n}\sum_{i=1}^q\max\left\{0,\log\sum_{j=1}^{\ell_i} |a_{i,j}|\right\},\,\,\,\,\beta > \beta_{v,n} = \frac{2}{n}\log_2\rho\bra{\overline{F}}
\end{align}
\end{subequations}
\end{thm}
\begin{proof}
 See Appendix \ref{sec: proofCuboidalVec}
\end{proof}

The thresholds if one uses an Ellipsoidal filter are given as follows.
\begin{thm}
\label{thm: ellipsoidalThmVec}
 It is possible to stabilize \eqref{eq: sysmodel} in the mean squared sense with an $(R,\beta)-$anytime code provided 
\begin{subequations}
\begin{align}
\label{eq: ellipsoidalThmVec}
 R > R_{ve,n} = \frac{1}{n}\sum_{i=1}^q\max\left\{0,\log\left[\sqrt{m_x}\sum_{j=1}^{\ell_i} |a_{i,j}|\theta^{j-1}\right]\right\},\,\,\,\,\beta > \beta_{ve,n} = \frac{2}{n}\log_2\rho\bra{F}
\end{align}
\end{subequations}
where $\theta = \sqrt{\frac{m_x}{m_x-1}}$\qedcustom
\end{thm}
We skip the proof for Theorem \ref{thm: ellipsoidalThmVec} since it is very similar to that of Theorem \ref{thm: cuboidalThmVec}.

%% file: Discussion.tex
\section{Discussion - Asymptotics and the Stabilizable Region}
\label{sec: discussion}
The sufficient conditions derived above are non-asymptotic in the sense that measurements are encoded every time step. Alternately, one can encode the measurements every, say, $\ell$ time steps, and consider the asymptotic rate and exponent needed as $\ell$ grows. This is often the form in which such sufficient conditions appear in the literature \cite{Sahai,Nair,Minero}. Even though the sufficient conditions in Sections \ref{sec: sufficient} and \ref{sec: vectorMeasurements} are non-asymptotic, note that they depend only on the system matrices $F$, $H$ and not on the noise distribution. In order to compare our results with those in the literature, we examine the sufficient conditions in the asymptotic limit of large $\ell$.

\subsection{The Limiting Case}
\label{subsec: theLimitingCase}
Note that encoding once every $\ell$ measurements amounts to working with the system matrix $F^\ell$. So, one can calculate this limiting rate and exponent by writing the eigen values of $F$, $\{\lambda_i\}_{i=1}^m$, as $\lambda_i = \mu_i^n$ and letting $n$ scale. The following asymptotic result allows us to compare the sufficient conditions above with those in the literature (eg., see \cite{Sahai,Nair,Minero}). 
\begin{thm}[The Limiting Case]
\label{thm: limitingCase}
Write the eigen values of $F$, $\{\lambda_i\}_{i=1}^{m_x}$, in the form $\lambda_i = \mu_i^{n}$. Letting $n$ scale, $R_{n}$, $R_{v,n}$, $R_{e,n}$, $R_{ev,n}$ converge to $R^*$, and $\beta_{n}$, $\beta_{v,n}$, $\beta_{e,n}$, $\beta_{ev,n}$ converge to $\beta^*$, where
\begin{align}
\label{eq: limiting} R^* =\sum_{i:|\mu_i|>1}\log_2|\mu_i|,\,\,\,\beta^* = 2\log_2\max_i|\mu_i|
\end{align}
\end{thm}
\begin{IEEEproof}
 See Appendix \ref{subsec: limitingCase}.
\end{IEEEproof}
For stabilizing plants over deterministic rate limited channels, \cite{Nair} showed that a rate $R > R^*$, where $R^*$ is as in \eqref{eq: limiting}, is necessary and sufficient. So, asymptotically the sufficient condition for the rate $R$ in Theorem \ref{thm: cuboidalThm} is tight. But it is not clear if one do with an exponent smaller than $\beta^* = 2\log_2\max_i|\mu_i|$ asymptotically when there is no feedback. Though the above limiting case allows one to obtain a tight and an intuitively pleasing characterization of the rate and exponent needed, it should be noted that this may not be operationally practical. For, if one encodes the measurements every $\ell$ time steps, even though Theorem \ref{thm: limitingCase} guarantees stability, the performance of the closed loop system (the LQR cost, say) may be unacceptably large because of the delay we incur. This is what motivated us to present the sufficient condition in the form that we did above. 

\subsection{A Comment on the Trade-off Between Rate and Exponent}

Once a set of rate-exponent pairs $(R,\beta)$ that can stabilize a plant is available, one would want to identify the pair that optimizes a given cost function. Higher rates provide finer resolution of the measurements while larger exponents ensure that the controller's estimate of the plant does not drift away; however, we cannot have both. One can either coarsely quantize the measurements and protect the bits heavily or quantize them moderately finely and not protect the bits as much. One can easily cook up examples using an LQR cost function with the balance going either way. Studying this trade-off is integral to making the results practically applicable.

\subsection{Stabilizable Region}

Using the thresholds obtained in Theorem \ref{thm: improvedOverRCE}, and the asymptotic sufficient condition in Theorem \ref{thm: limitingCase}, we can discuss the range of the eigen values of $F$, i.e., $\{|\mu_i|\}_{i=1}^{m_x}$, for which the $\eta^{th}$ moment of $x_t$ in \eqref{eq: sysmodel} can be stabilized over some common channels. Since we are interested in the asymptotics, we assume the same limiting case as in Section \ref{subsec: theLimitingCase}. Firstly, consider the scalar case, i.e., $m_x = 1$ and let the eigen value be $\mu$. An anytime reliable code with rate $R$ and exponent $\beta$ can stabilize the process in \eqref{eq: sysmodel} for all $\mu$ such that 
\begin{align*}
 \log_2|\mu| < \min\left\{R,\frac{\beta}{\eta}\right\}
\end{align*}
So, a scalar unstable linear process in \eqref{eq: sysmodel} can be stabilized over a MBIOS channel with Bhattacharya parameter $\bh$ provided
\begin{align}
\label{eq: region1}
 \log_2|\mu| < \log_2|\mu_{\max}| = \sup_{R < C,\beta < E_\bh(R)}\min\left\{R,\frac{\beta}{\eta}\right\}
\end{align}
The stabilizable region as implied by the threshold in \cite{Sahai} is given by 
\begin{align*} 
\log_2|\mu| < \log_2|\mu_{\max}| = \sup_{R < C,\beta < E_r(R)}\min\left\{R,\frac{\beta}{\eta}\right\}
\end{align*}
For $\eta=2$, the stabilizable region for the BEC and BSC is shown in Fig \ref{fig: BECvBSC} where $|\mu_{max}|$ is plotted against the channel parameter.
\begin{figure}
 \centering
\includegraphics[scale=0.25]{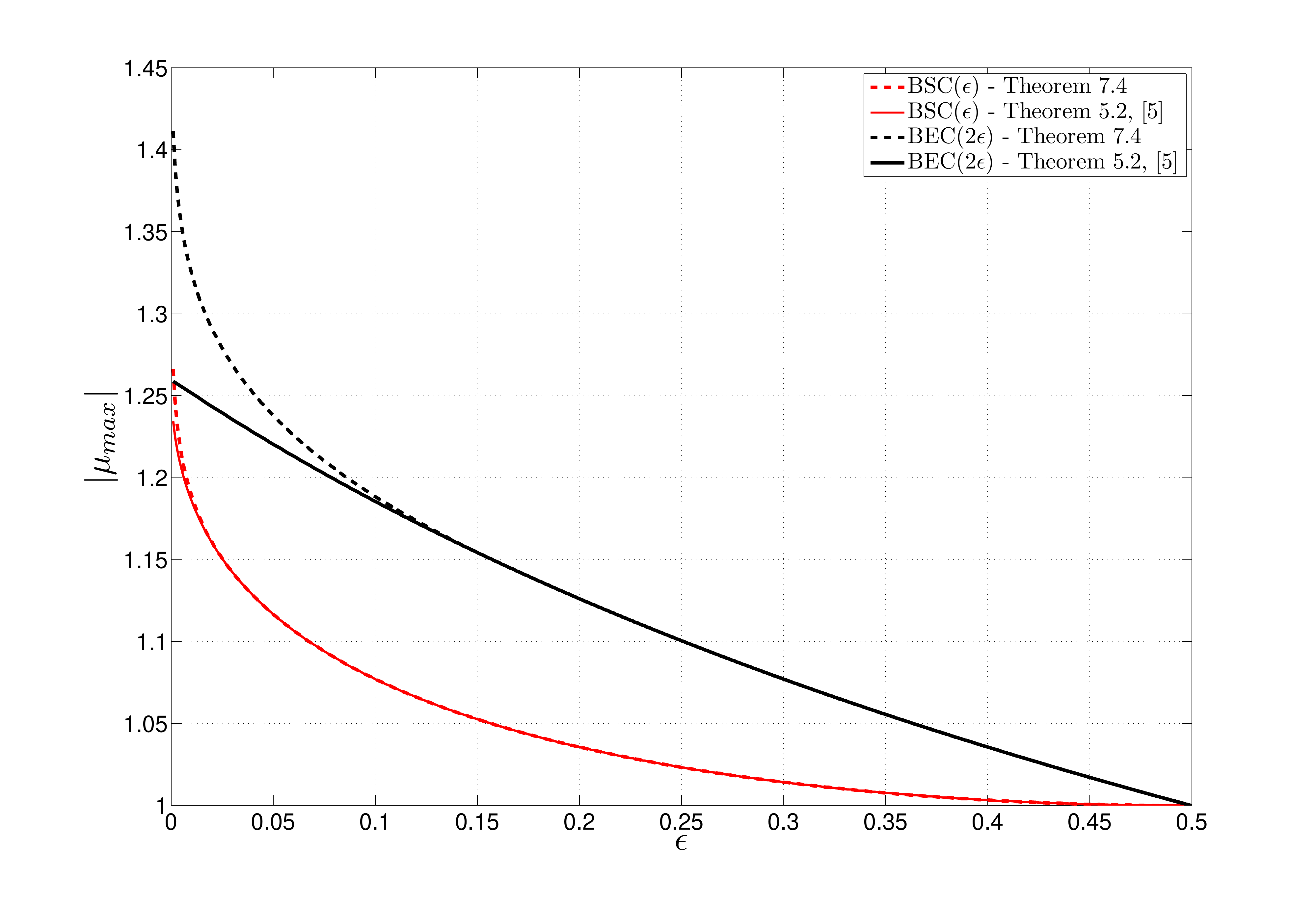}
\caption{Comparing the stabilizable regions of BSC and BEC using linear codes}
\label{fig: BECvBSC}
\end{figure}
Consider a vector valued process with unstable eigen values $\{|\mu_i|\}_{i=1}^{m}$. Such a process can be stabilized by a rate $R$ and exponent $\beta$ anytime reliable code provided $R > \sum_{i=1}^m\log|\mu_i|$ and $\beta > \log\bra{\max_{i}|\mu_i|}$. So, given a channel with Bhattacharya parameter $\bh$ for which the rate exponent curve $(R,E_{\bh}(R))$ is achievable, the region of unstable eigen values that can be stabilized is given by $\{\mu\in\Re^m,\,\,|\,\,\,\exists R < C\ni \sum_{i=1}^m\log|\mu_i| < R\,\,\text{ and }\,\,\log\bra{\max_{i}|\mu_i|} < E_{\bh}(R)\}$, where $C$ is the Shannon capacity of the channel. For example, let $m=2$ and $\eta=2$. Fig \ref{fig: 2DimStabRegion}a shows the region of $(|\mu_1|,|\mu_2|)$ that can be stabilized over three different channels, a binary symmetric channel with bit flip probability 0.1 and binary erasure channels with erasure probabilities 0.1 and 0.2 respectively.
\begin{figure}
 \centering
\subfigure[Each curve represents the outer boundary of the stabilizable region.]{\includegraphics[scale=0.26]{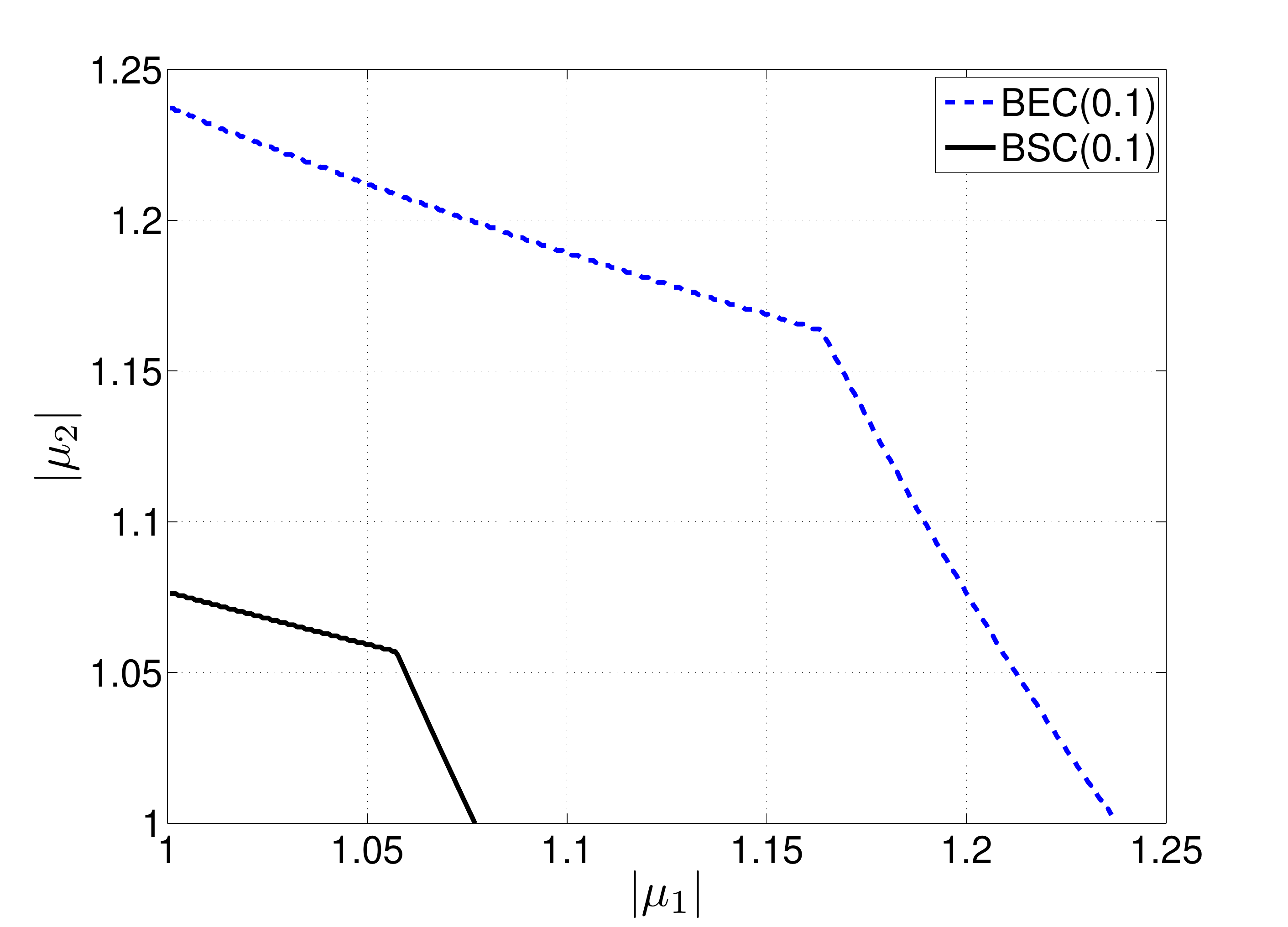}}\,\,
\subfigure[Stabilizable region with and without feedback]{\includegraphics[scale=0.25]{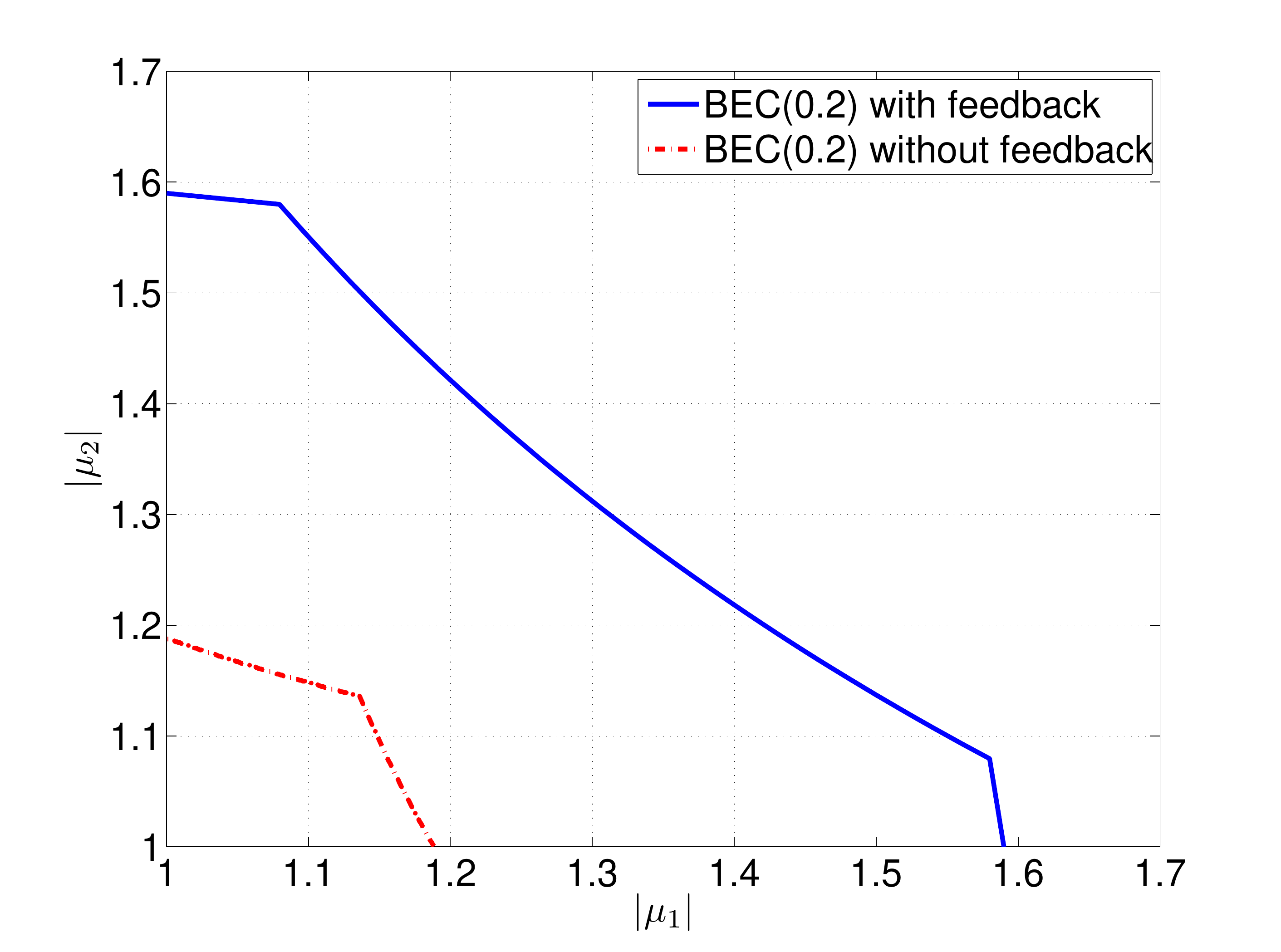}}
\caption{Comparing the stabilizable region of different channels}
\label{fig: 2DimStabRegion}
\end{figure}

We will now compare these results with the case when there is perfect feedback of the channel outputs at the observer/encoder. \cite{SahaiVec} considered a priority queuing method for stabilizing vector valued unstable processes over channels with perfect feedback. Bits from different unstable subsystems are placed in a FIFO queue. Bits are given preference in decreasing order of the size of the eigen value of the corresponding subsystem. So, bits coming from a subsystem with a larger eigen value are given preference over those from a subsystem with a smaller eigen value. A bit is removed from the queue once it is received correctly. Since the feedback anytime capacity of a binary erasure channel is known \cite{SahaiWhy}, one can use Theorem 6.1 in \cite{SahaiVec} to derive the region of eigen values that can be stabilized by such a scheme. In Fig. \ref{fig: 2DimStabRegion}b, we compare the region of $(|\mu_1|,|\mu_2|)$ that can be stabilized with and without feedback over a binary erasure channel with erasure probability 0.2. As one would expect, the region is much larger when there is feedback. Note that the stabilizable regions in Fig. \ref{fig: 2DimStabRegion} are only achievable and not necessarily tight.

%% file: Simulations.tex
\section{Simulations}
\label{sec: Simulations}
We present two examples and stabilize them over a binary erasure channel with erasure probability $\epsilon = 0.3$. The number of channel uses per measurement is fixed to $n=15$. In both cases, time invariant codes $\mathbb{H}_{15,R}\in\mathbb{TZ}_{\frac{1}{2}}$, for an appropriate rate $R$, were randomly generated and decoded using Algorithm \ref{alg: algorithm}. The controller uses the Hypercuboidal filter to estimate the state.
\subsection{Cart-Stick Balancer}
The system parameters for a cart-stick balancer (also commonly called the \textit{inverted pendulum on a cart}) with state variables of stick angle, stick angular velocity, and cart velocity, when sampled with sampling duration 0.1s are (Exercise 10.15 in \cite{Franklin})
\begin{align*}
F = \left[\begin{array}{ccc}
                  1.161& 0.105& 0\\
		  3.3& 1.161& 0.002\\
		  -3.265 & -0.160& 0.979 
                 \end{array}
\right],\,\,\,G = [-0.003\,\,\,-0.068\,\,\,0.859]^T,\,\,\,H = [10\,\,\,0\,\,\,0]
\end{align*}
The characteristic polynomial of $F$ is $x^3 - 3.3x^2 + 3.27x - 0.98$ and its eigen values are 1.75, 0.98 and 0.57. So, $F$ is open loop unstable. Each component of the process noise and measurement noise is i.i.d zero mean Gaussian with variance 0.01 truncated to lie in [-0.025,0.025]. The control input is given by $u_t = -K\hat{x}_{t|t}$, where $K = [-81.55\,\,\,-14.37\,\,\,-0.04]$. One can verify that $F-GK$ is stable. 
In order to apply Theorem \ref{thm: cuboidalThm}, we write $F$ in the following canonical form
\begin{align*}
 F_o = \left[\begin{array}{ccc}
                  3.3& 1& 0\\
		  -3.27& 0& 1\\
		  0.98 & 0& 0 
                 \end{array}
\right]
\end{align*}
Applying Theorem \ref{thm: cuboidalThm}, one can stabilize $x_t$ in the mean squared sense provided the exponent $n\beta > 2\log\bra{\rho\bra{\overline{F_o}}} = 4.1035$ and the rate $nR = k > \log\bra{3.3+3.27+0.98}=2.1$. For $k=5$, there exist anytime reliable codes with exponent upto $n\beta=4.27$. Fig \ref{fig: samplePath} plots a sample path of the above system for a randomly chosen Toeplitz code. It is clear from Fig \ref{fig: samplePathNorm} that the plant is stabilized.

\begin{figure}
\centering
\subfigure[The stick does not deviate by more than 3 degrees from the vertical]{\includegraphics[height=2in,width=3.3in]{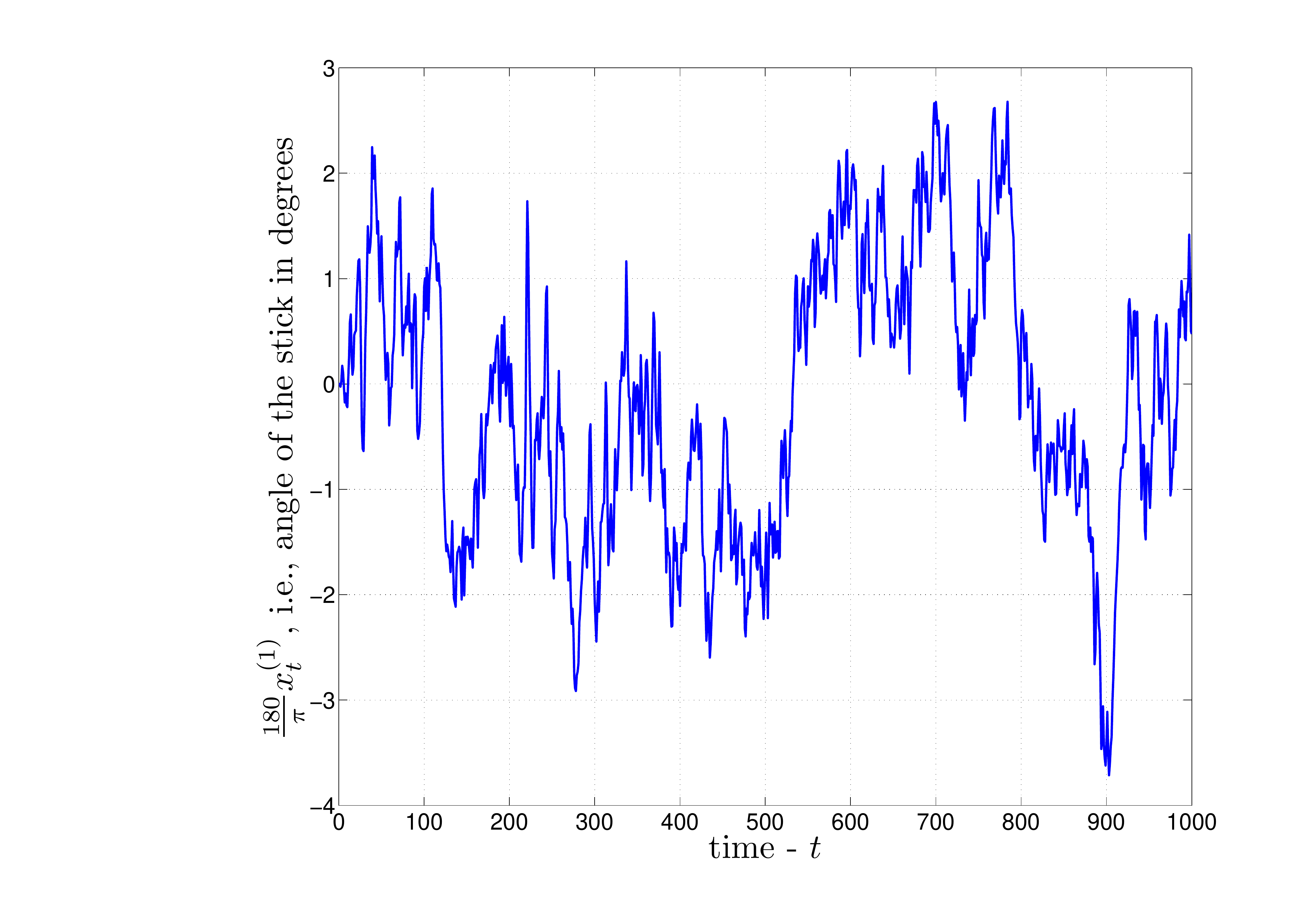}\label{fig: samplePathAngle}}\,\,
\subfigure[This shows that the plant is stabilized]{\includegraphics[height=2in,width=2.8in]{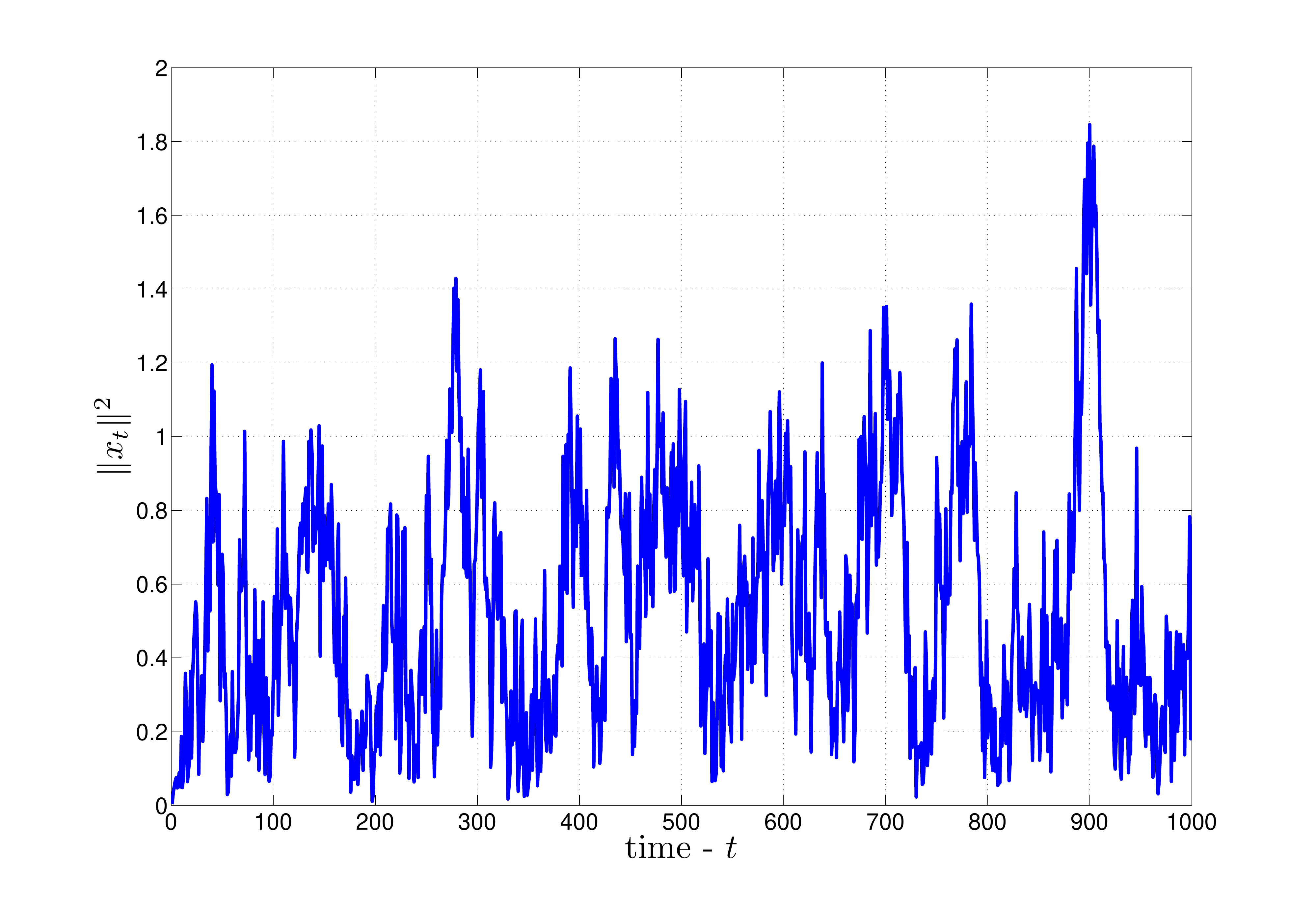}\label{fig: samplePathNorm}}
\caption{A sample path}
\label{fig: samplePath}
\end{figure} 

\subsection{Example 2}
This example is aimed at exploring the trade off between the resolution of the quantizer and the error performance of the causal code. Consider a 3-dimensional unstable system \eqref{eq: sysmodel} with 
\begin{align*}
F = \left[\begin{array}{ccc}2 & 1 & 0\\0.25 & 0 & 1\\-0.5 & 0 & 0\end{array}\right] 
\end{align*}
$G = \mathcal{I}_3$ and $H=[1 0 0]$. Each component of $w_t$ and $v_t$ is generated i.i.d $N(0,1)$ and truncated to [-2.5,2.5]. The eigen values of $F$ are $\{2,-0.5,0.5\}$ while $\lambda(\Fbar) = 2.215$. The observer has access to the control inputs and we use the hypercuboidal filter outlined in Section \ref{sec: proofCuboidal}. Using Theorem \ref{thm: cuboidalThm}, the minimum required bits and exponent are given by $k = nR \geq 2$ and $n\beta\geq 2\log_2 2.215 = 2.29$. The control input is $u_t = -\hat{x}_{t|t-1}$. For $k\leq 7$, $n\beta \geq 2.32$. If $k = 8$, $n\beta = 1.32 < 2.29$. For each value of $k$ ranging from 3 to 7, 1000 codes were generated from the ensemble $\mathbb{TZ}_{\frac{1}{2}}$. For each code, the system was simulated over a horizon of 100 time instants and the LQR cost has been averaged over 100 such runs. For a time horizon $T$, the LQR cost is defined as $\frac{1}{2T}\sum_{t=0}^T\E\bra{\|x_t\|^2 + \|u_t\|^2}$. In Fig \ref{fig: 3DimSysCDF}, the cumulative distribution function of the LQR cost is plotted for $3\leq k\leq 7$. The $x-$axis denotes the proportion of codes for which the LQR cost is below a prescribed value, e.g., with $k=6,n=15$, the cost was less than 15 for $85\%$ of the codes while with $k=5,n=15$, this fraction increases to more than $95\%$. 
The competition between the rate and the exponent in determining the LQR cost is evident when we look at Fig \ref{fig: 3DimSysLQR}. When $k = 3$, the error exponent $n\beta = 6.3$ is large. So, at any time $t$, the decoder decodes all the source bits $\{b_\tau\}_{\tau\leq t-1}$ with a high probability. Hence, the limiting factor on the LQR cost is the resolution that the source bits $b_t$ provide on the measurements. But when $k = 7$, the measurements are quantized to a high resolution but the decoder makes errors in decoding the source bits. So, the best choice appears to be $k=5$.

\begin{figure}
\centering
\subfigure[The CDF of the LQR costs for different values of the rate]{\includegraphics[scale=0.23]{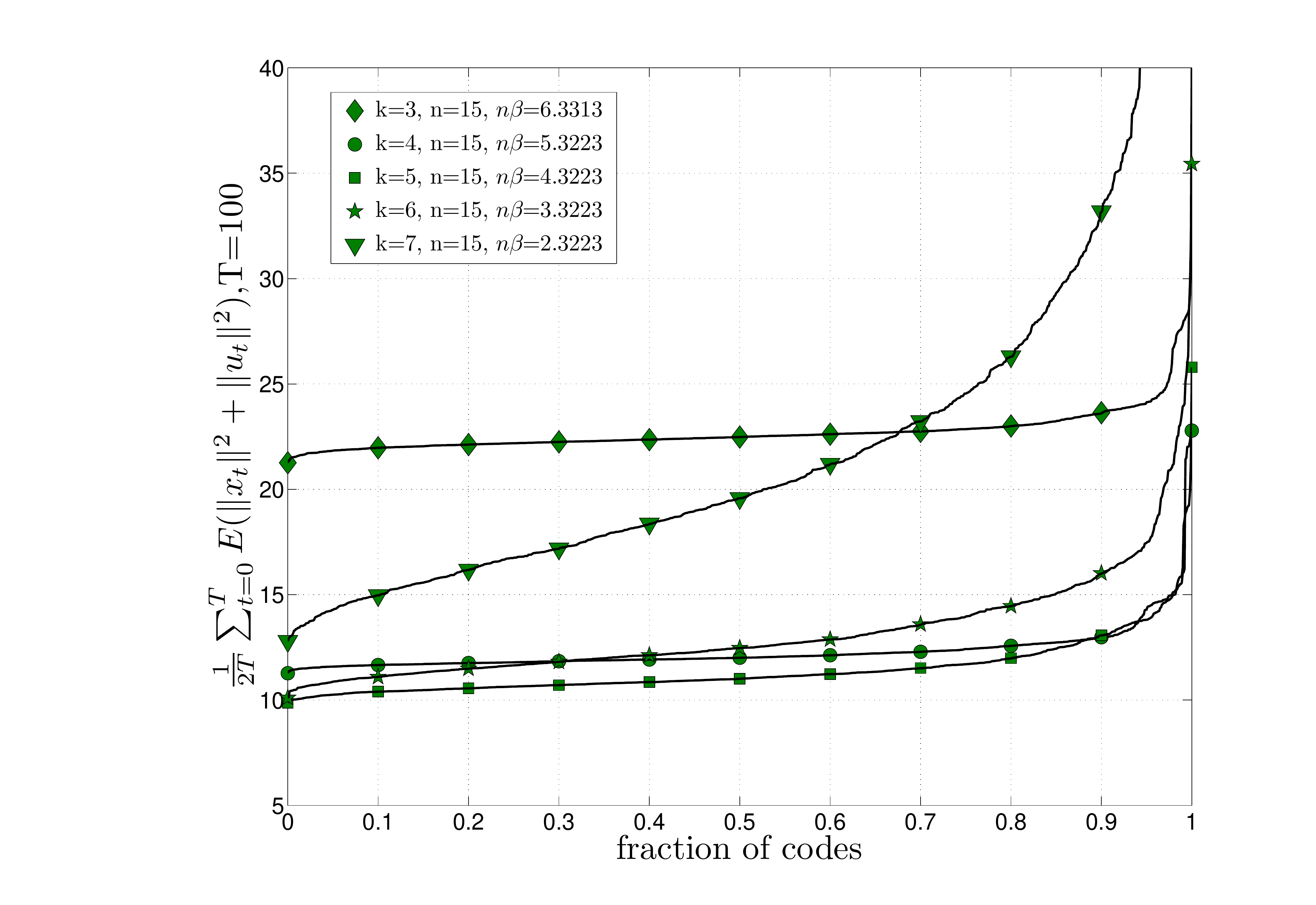}\label{fig: 3DimSysCDF}}\,\,
 \subfigure[The LQR cost averaged over the 1000 randomly generated codes is plotted against $k$]{\includegraphics[scale=0.15]{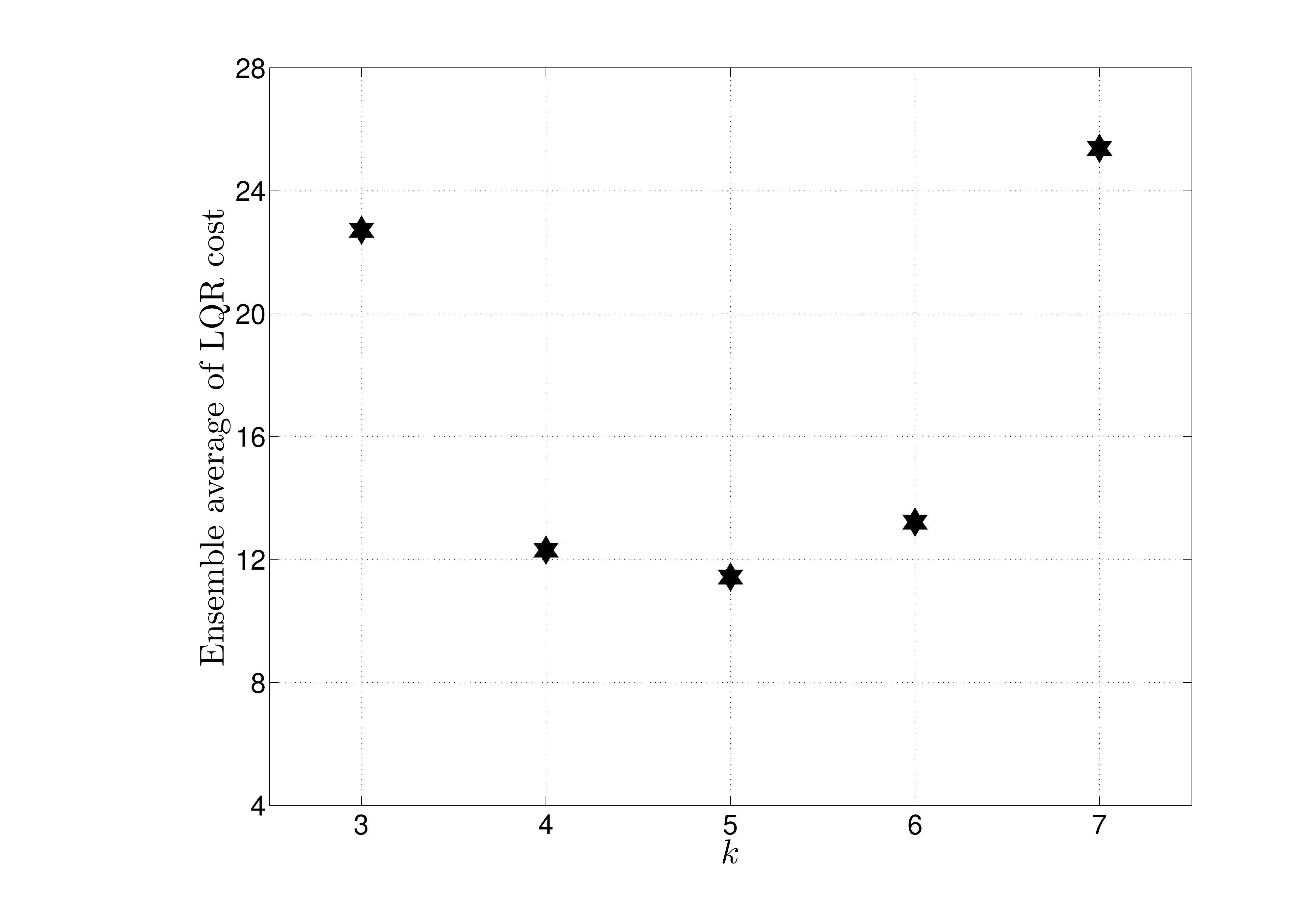}\label{fig: 3DimSysLQR}}
\caption{The best choice of the rate is $R = 5/15 = 0.33$}
\end{figure}

%% file: Appendix.tex
\appendix

\input{AppendixProofToeplitz}
\input{AppendixProofCuboidal}

\subsection{The Minimum Volume Ellipsoid}
\begin{lem}[Theorem 6.1 \cite{Guler}]
\label{lem: minVol}
 The minimum volume ellipsoid $\mE(\hat{P},c)$ covering 
\begin{align*}
\left\{x\in\Re^m| x\in\mE(P,0),\gamma\sqrt{h^TPh}\leq \langle h,x\rangle\leq\delta\sqrt{h^TPh}\right\} 
\end{align*}
where $|\delta| \geq |\gamma|$, is given by
\begin{align}
 \hat{P} = bP - (b-a)\frac{Phh^TP}{h^TPh},\,\,\,c = \xi \frac{Ph}{\sqrt{h^TPh}}
\end{align}
where 
\begin{enumerate}
 \item If $\gamma\delta < -\frac{1}{m}$, then $\xi = 0$, $a = b = 1$
 \item If $\gamma+\delta = 0$ and $\gamma\delta > -\frac{1}{m}$, then
\begin{align*}
 \xi = 0,\,\,a = m\delta^2,\,\,b = \frac{m(1-\delta^2)}{m-1}
\end{align*}
 \item If $\gamma+\delta\neq 0$ and $\gamma\delta > -\frac{1}{m}$, then
\begin{align*}
 \xi &= \frac{m(\gamma+\delta)^2 + 2(1+\gamma\delta) - \sqrt{D}}{2(m+1)(\gamma+\delta)}\\
a &= m(\xi-\gamma)(\delta-\xi),\,\,b = \frac{a-a\gamma^2}{a-(\xi-\gamma)^2}\\
\text{where }D&= m^2(\delta^2 - \gamma^2)^2 + 4(1-\gamma^2)(1-\delta^2)
\end{align*}
\end{enumerate}
\end{lem}
If $|\delta| < |\gamma|$, change $x$ to $-x$ and apply the above result. And it is easy to verify that $\hat{P}$ is indeed positive semidefinite. Also, a quick calculation shows that $\gamma \leq \xi\leq \delta$. This confirms the intuition that the center of the minimum volume ellipsoid lies within the slab. 

\input{AppendixProofEllipsoidal}

\input{AppendixLimitingCase}

%% file: AppendixProofToeplitz.tex
\subsection{Proof of Theorem \ref{thm: Toeplitz}}
\label{sec: proofToeplitz}

We will begin with some preliminary observations.
\begin{lem}[\cite{Kakhaki}]
 \label{lem: subspace}
Let $V$ be an $m-$dimensional vector space over $\GF_2$ and define a probability function over $V$ such that, for each $v\in V$, $P(v) = p^{\|v\|}(1-p)^{m-\|v\|}$.
If $U$ is an $\ell-$dimensional subspace of $V$, then
\begin{align*}
 P(U) \leq \max(p,1-p)^{m-\ell}
\end{align*}
\end{lem}
\begin{proof}
Suppose $p\leq 1/2$. The proof for the other case is analogous. Let $E$ be the set of unit vectors, i.e., $E = \left\{v\in V\,|\,\|v\|=1\right\}$. Then there is a subset, $E'$, of $E$ with $m-\ell$ unit vectors such that $V = U\oplus span(E')$ and $U\cap span(E') = \{0\}$. Let $u'\in span(E')$, then 
\begin{align*}
 P\bra{U + u'} = \sum_{u\in U}P(u+u') \geq \sum_{u\in U}P(u)\bra{\frac{p}{1-p}}^{\|u'\|} = P(U)\bra{\frac{p}{1-p}}^{\|u'\|}
\end{align*}
Note that for distinct $u_1',u_2'\in span(E')$, $(U+u_1')\cap(U + u_2') = \emptyset$. Also note that $\|u'\| \leq m-\ell$ $\forall$ $u'\in span(E')$.
\begin{align*}
 1 = P(V) = P\bra{\bigcup_{u'\in span(E')}(U+u')} \geq \sum_{u'\in span(E')}P(U)\bra{\frac{p}{1-p}}^{\|u'\|}
\end{align*}
Observe that there are exactly $\binom{m-\ell}{i}$ vectors in $span(E')$ with Hamming weight $i$. So, we have
\begin{align*}
 1 \geq P(U)\sum_{i=0}^{m-\ell}\binom{m-\ell}{i}\bra{\frac{p}{1-p}}^{i} = P(U)\bra{\frac{1}{1-p}}^{m-\ell}
\end{align*}
This completes the proof.
\end{proof}

\begin{rem}
\label{rem: rank}
The Toeplitz parity check matrix $\mathbb{H}_{n,R}^{TZ}$ is full rank if and only if $H_1$ is full rank. This is why we fix $H_1$ to be a full rank matrix in the definition of the Toeplitz ensemble.
\end{rem}

Recall that we choose the entries of $H_i$ to be i.i.d Bernoulli($p$) for $i \geq 2$. Also suppose $p \leq 1/2$. The results for $p\geq 1/2$ are obtained by replacing $p$ with $1-p$ in the subsequent analysis. Consider an arbitrary decoding instant, $t$. Since $\minW{d}{t} = \minW{d}{t'}$ and $\numW{w}{d}{t} = \numW{w}{d}{t'}$ for all $t,t'$, we will drop these superscripts and write $\minW{d}{t} = \minWi{d}$ and $\numW{w}{d}{t} = \numWi{w}{d}$. Let $c = [c_1^T,\ldots,c_t^T]^T$, where $c_i\in\{0,1\}^n$, be a fixed binary word such that $c_{\tau < t-d+1}=0$ and $c_{t-d+1}\neq 0$. Also, let $\mathbb{H}_{n,R}$ be drawn from the ensemble $\mathbb{TZ}_p$ and let $\mathbb{H}_{n,R}^t$ denote the $\overline{n}t\times nt$ principal minor of $\mathbb{H}_{n,R}$. We examine the probability that $c$ is a codeword of $\mathbb{H}_{n,R}^t$, i.e., $P\bra{\mathbb{H}_{n,R}^tc=0}$. Now, since $c_{\tau<t-d+1}=0$, $\mathbb{H}_{n,R}^tc=0$ is equivalent to
\begin{align}
\label{eq: toeplitzproof1}
\left[\begin{array}{cccc}
	\pc_1    & 0 	    & \ldots      &  \ldots\\
	\pc_2    & \pc_1         & 0 	     & \ldots \\
	\vdots & \vdots      & \ddots     & \vdots\\
	\pc_d     & \pc_{d-1} & \ldots       & \pc_1
\end{array}\right]\left[\begin{array}{c}c_{t-d+1}\\c_{t-d+2}\\\vdots\\c_t\end{array}\right]=\left[\begin{array}{c}0\\0\\\vdots\\0\end{array}\right]
\end{align}
Note that \eqref{eq: toeplitzproof1} can be equivalently written as follows
\begin{align}
\label{eq: cute1}
\left[\begin{array}{cccc}
	C_{t-d+1}    & 0 	    & \ldots      &  \ldots\\
	C_{t-d+2}   & C_{t-d+1}       & 0 	     & \ldots \\
	\vdots & \vdots      & \ddots     & \vdots\\
	C_t   & C_{t-1} & \ldots       & C_{t-d+1}
\end{array}\right]\left[\begin{array}{c}h_1\\h_2\\\vdots\\h_d\end{array}\right]=\left[\begin{array}{c}0\\0\\\vdots\\0\end{array}\right]
\end{align}
where $h_i = \text{vec}(\pc_i^T)$, i.e., $h_i$ is a $n\overline{n}\times 1$ column obtained by stacking the columns of $\pc_i^T$ one below the other, and $C_i \in\{0,1\}^{\overline{n}\times n\overline{n}}$ is obtained from $c_i$ as follows. 
\begin{align}
\label{eq: cute2}
 C_i =  \left[\begin{array}{cccc}
	c_i^T    & 0 	    & \ldots      &  \ldots\\
	0   & c_i^T       & 0 	     & \ldots \\
	\vdots & \vdots      & \ddots     & \vdots\\
	0  & 0 & \ldots       & c_i^T
\end{array}\right]
\end{align}
Since $H_1$ is fixed, we will rewrite \eqref{eq: cute1} as 
\begin{align}
 \label{eq: cute3}
\underbrace{ \left[\begin{array}{cccc}
	C_{t-d+1}    & 0 	    & \ldots      &  \ldots\\
	C_{t-d+2}   & C_{t-d+1}       & 0 	     & \ldots \\
	\vdots & \vdots      & \ddots     & \vdots\\
	C_{t-1}   & C_{t-2} & \ldots       & C_{t-d+1}
\end{array}\right]}_{\triangleq C}\underbrace{\left[\begin{array}{c}h_2\\h_3\\\vdots\\h_d\end{array}\right]}_{\triangleq h} = \left[\begin{array}{c}C_{t-d+2}\\C_{t-d+3}\\\vdots\\C_t\end{array}\right]h_1,\,\,\,C_{t-d+1}h_1 = 0
\end{align}
Since $c_{t-d+1} \neq 0$, $C_{t-d+1}$ has full rank $\overline{n}$ and consequently $C$ has full rank $(d-1)\overline{n}$. Since $C$ is an $(d-1)\overline{n}\times(d-1)n\overline{n}$ matrix, its null space has dimension $(d-1)n\overline{n} - (d-1)\overline{n}$. For \eqref{eq: cute3} to hold, $h$ must lie in an $(d-1)n\overline{n} - (d-1)\overline{n}$ dimensional flat which is contained in an $(d-1)n\overline{n} - (d-1)\overline{n} + 1$ dimensional subspace. Using Lemma \ref{lem: subspace}, we have
\begin{align}
\label{eq: toeplitzproof2_5}P(\mathbb{H}_{n,R}^tc = 0) &\leq (1-p)^{\overline{n}(d-1)-1}\\
\implies P\bra{\minWi{d} < \alpha nd} &\leq (1-p)^{\overline{n}(d-1)-1}\sum_{w'\leq \alpha nd}\binom{nd}{w'}\nonumber\\
 &\leq (1-p)^{\overline{n}(d-1)-1}2^{ndH(\alpha)}\nonumber\\
\label{eq: toeplitzproof2}& = \eta2^{-nd\bra{(1-R)\log_2(1/(1-p)) - H(\alpha)}}
\end{align}
where $\eta = (1-p)^{-\overline{n}-1}$. Similarly,
\begin{align}
 \label{eq: teoplitzproof3}
P\bra{\numWi{w}{d} > 2^{\theta w}} &\leq 2^{-\theta w}\E\numWi{w}{d}\nonumber\\
							      &\leq  \eta2^{-\theta w}\binom{nd}{w}(1-p)^{\overline{n}d}\nonumber\\
							      &\leq \eta2^{-nd\bra{\theta w/nd - H(w/nd) + (1-R)\log_2(1/(1-p))}}
\end{align}
For convenience, define
\begin{align*}
 \delta_1 &= (1-R)\log_2(1/(1-p)) - H(\alpha)\\
 \delta_{2,w} &= \theta \frac{w}{nd} - H\bra{\frac{w}{nd}} + (1-R)\log_2(1/(1-p))
\end{align*}
We need to choose $\theta$ such that $\delta_{2,w} > \delta > 0$ for all $\alpha \leq \frac{w}{nd} \leq 1$. Now, define
\begin{align}
 \theta^* = \max_{x\geq \alpha}\frac{H(x) - (1-R)}{x}
\end{align}
Then for each $\theta > \theta^*$, there is a $\delta > 0$ such that $\delta_{2,w} > \delta$ for all $\alpha nd \leq w \leq nd$. A simple calculation gives $\theta^* = \log_2\bra{\frac{1}{2^{1-R}-1}}$. For such a choice of $\theta > \theta^*$, continuing from \eqref{eq: teoplitzproof3}, we have
\begin{align}
\label{eq: toeplitzproof4}
 P\bra{\exists\,\alpha nd\leq w\leq nd\,\,\ni\,\,\numWi{w}{d} > 2^{\theta w}} \leq nd2^{-nd\delta}
\end{align}
for some $\delta' > 0$. For some fixed $d_o$ large enough, applying a union bound over $d\geq d_o$ to \eqref{eq: toeplitzproof2} and \eqref{eq: toeplitzproof4}, we get
\begin{align}
 \label{eq: toeplitzproof5}
P\bra{\exists\,\,d\geq d_o\,\,\ni\,\,\minWi{d}<\alpha nd\text{ or }\numWi{w}{d}>2^{\theta w}} \leq 2^{-\Omega(nd_o)}
\end{align}

%% file: AppendixProofCuboidal.tex
\def\tFo{\tilde{F}_o}
\def\tHo{\tilde{H}_o}
\subsection{Proofs of Theorems \ref{thm: cuboidalThm} and \ref{thm: cuboidalThmVec}}
\label{sec: proofCuboidal}

\subsubsection{Proof of Theorem \ref{thm: cuboidalThm}}
The analysis will proceed in two steps. We will first determine a sufficient condition on the number of bits per measurement, $nR$, that are required to track \eqref{eq: sysmodel} when these bits are available error free. We will then determine the anytime exponent $n\beta$ needed in decoding these source bits when they are communicated over a noisy channel. 

Let $\Delta_{t|\tau}\triangleq x_{max,t|\tau}-x_{min,t|\tau}$ be the uncertainty in $x_t$ using $\{b_{\tau'}\}_{\tau'\leq \tau}$, i.e., quantized measurements up to time $\tau$. For convenience, let $\Delta_t \equiv \Delta_{t|t-1}$. Then, the time update is given by the following Lemma.
\begin{lem}[Time Update]
\label{lem: cubTimeUpdate}
The time update relating $\Delta_{t+1}$ and $\Delta_{t|t}$ is given by $\Delta_{t+1} = \Fbar\Delta_{t|t} + W\mones$
\end{lem}
\begin{IEEEproof}
 From the system dynamics in \eqref{eq: sysmodel}, the following is immediate
\begin{align*}
 \Delta_{t+1}^{(i)} &= W + \max\left\{\left|\pm a_i\Delta_{t|t}^{(1)} + \Delta_{t|t}^{(i+1)}\right|,\left|\Delta_{t|t}^{(i+1)}\right|,\left|a_i\Delta_{t|t}^{(1)}\right|\right\}\\
 &= |a_i|\Delta_{t|t}^{(1)} + \Delta_{t|t}^{(i+1)} + W,\,\,\,i\leq m-1\\
\Delta_{t+1}^{(m)} &= |a_m|\Delta_{t|t}^{(1)} + W
\end{align*}
In short, the above equations amount to $\Delta_{t+1} = \Fbar\Delta_{t|t} + W\mones$.
\end{IEEEproof}

Towards the measurement update, the observer simply quantizes the measurements $y_t$ according to a $2^{nR}-$regular lattice quantizer with bin width $\delta$, i.e., the quantizer is defined by $Q:\Re\mapsto\{0,1,\ldots,2^{nR}-1\}$, where $Q(x) = \lfloor\frac{x}{\delta}\rfloor\text{ mod }2^{nR}$. In order for this to work, we need $\delta 2^{nR} \geq \Delta_t^{(1)}$ for any time $t$. Assuming that the rate, $R$, is large enough, we will first find the steady state value of the recursion for $\Delta_t$, which we then use to determine $R$. At each time $t$, the observer can communicate the measurement $y_t$ to within an uncertainty of $\delta$, i.e., the estimator knows that the measurement lies in an interval of width $\delta$. Adding to this the effect of the observation noise, $-\frac{V}{2} \leq v_t \leq \frac{V}{2}$, the estimator knows $x_t^{(1)}$ to within an uncertainty of $\Delta_{t|t}^{(1)} = \delta + V$. Note that $\Delta_{t|t}^{(i)} = \Delta_{t}^{(i)}$ for $i\neq 1$. Combining this observation with Lemma \ref{lem: cubTimeUpdate}, it is straightforward to see that $\Delta_{t}$ converges, to say $\Delta_{tu}$, in exactly $m_x$ time steps, i.e., $\Delta_{t} = \Delta_{tu}$ for all $t \geq m_x$. The subscript \lq tu\rq\phantom{} in $\Delta_{tu}$ denotes \lq time update\rq\phantom{}. The following result is now immediate.
\begin{lem}[Steady State value of $\Delta_t$]
\label{lem: determiningDelta}
$\Delta_{tu} = (\delta+V)L_ua + WL_u\mones$, where $a = [|a_1|,\ldots,|a_m|]^T$ and $L_u = [\ell_{ij}]_{1\leq i,j\leq m}$ with $\ell_{ij} = \mathbb{I}_{i\leq j}$.
\end{lem}

Now, we need to go back and calculate $R$. So we just need $\delta2^{nR}\geq \max\left\{\Delta_0^{(1)},\Delta_1^{(1)},\ldots,\Delta_{m_x}^{(1)}\right\}$. Further, a simple calculation gives $\lim_{\delta\rightarrow\infty}\frac{\Delta_i^{(1)}}{\delta} = |a_1| + \ldots + |a_i|$. The minimum rate is thus given by $\frac{1}{n}\log_2\sum_{i=1}^m|a_i|$ and this completes the proof Theorem \ref{thm: cuboidalThm}. 

\subsubsection{Proof of Theorem \ref{thm: cuboidalThmVec}}
\label{sec: proofCuboidalVec}

The proof is very similar to that of Theorem \ref{thm: cuboidalThm}. The observations are quantized as follows. At any time, for $1\leq i\leq q$, the $i^{th}$ component of the measurement vector is quantized using a $2^{nR_i}-$regular lattice quantizer with bin width $\delta_i$. The remaining components of the measurement vector are ignored. The overall rate, $R$, is then given by $R = R_1+R_2\ldots+R_q$. The time update again is given by $\Delta_{t+1} = \overline{F}\Delta_{t|t} + W\mones$. The limiting values of $\{R_i\}_{i=1}^q$ are obtained by letting $\delta_1\rightarrow\infty$ and $\frac{\delta_i}{\delta_{i+1}}\rightarrow\infty$. An argument similar to the one in the previous section gives the following threshold, $R_i \geq \frac{1}{n}\max\left\{0,\log\bra{|a_{i,1}|+|a_{i,2}|+\ldots+|a_{i,\ell_i}|}\right\}$.

%% file: AppendixProofEllipsoidal.tex
\subsection{Proof of Theorem \ref{thm: ellipsoidalThm}}
\label{sec: proofEllipsoidal}

The proof is in the same spirit as that of Theorem \ref{thm: cuboidalThm}. We will first determine a sufficient condition on the number of bits per measurement, $nR$, that are required to track \eqref{eq: sysmodel} when these bits are available error free. We will then determine the anytime exponent $n\beta$ needed in decoding these source bits when they are communicated over a noisy channel. 

Consider the time update in \eqref{eq: ellipTU}. Let $P_t^{ij}$ denote the $(i,j)^{th}$ element of $P_t$, then the time update implies
\begin{subequations}
\begin{align}
\label{eq: proofEllip1} P_{t+1}^{ii} &= (1+\epsilon')\bra{a_i^2P_{t|t}^{11} + P_{t|t}^{i+1,i+1} - a_iP_{t|t}^{1,i+1} - a_iP_{t|t}^{i+1,1}} + \frac{W^2}{4\epsilon'},\,\,1\leq i\leq m_x-1\\
 P_{t+1}^{m_x,m_x} &= (1+\epsilon')a_{m_x}^2P_{t|t}^{11} + \frac{W^2}{4\epsilon'}
\end{align}
\label{eq: proofEllipTU}
\end{subequations}
Since the matrix $P_{t|t}$ is positive semidefinite, we have $P_{t|t}^{1,i+1} = P_{t|t}^{i+1,1}$ and $\bra{P_{t|t}^{1,i+1}}^2 \leq P_{t|t}^{11}P_{t|t}^{i+1,i+1}$. Using this in \eqref{eq: proofEllip1}, for $1\leq i\leq m_x-1$, we get
\begin{align}
 \label{eq: proofEllip0}P_{t+1}^{ii} \leq (1+\epsilon')\bra{|a_i|\sqrt{P_{t|t}^{11}} + \sqrt{P_{t|t}^{i+1,i+1}}}^2 + \frac{W^2}{4\epsilon'}
\end{align}
This prompts us to bound the recursion \eqref{eq: ellip} by bounding the diagonal elements of $P_t$. Now, considering the measurement update \eqref{eq: ellipMU}, it is easy to see that
\begin{subequations}
 \begin{align}
\label{eq: proofEllip2} P_{t|t}^{11} &= a_tP_t^{11}\\
a_tP_t^{ii} \leq P_{t|t}^{ii} &\leq b_tP_t^{ii}
\end{align}
\end{subequations}
We will first show that $b_t \leq \frac{m_x}{m_x-1}$. 
\begin{lem}
 \label{lem: boundbt}
$b_t \leq \frac{m_x}{m_x-1}$
\end{lem}
\begin{proof}
 To prove this, consider the setup of Lemma \ref{lem: minVol} and suppose $|\delta|\geq|\gamma|$. Then, in cases 1) and 2), it is clear that $b\leq \frac{m}{m-1}$ since $|\delta|,|\gamma| \leq 1$. In case 3), we have 
\begin{align*}
 b = \frac{1-\gamma^2}{1-(\xi-\gamma)^2/a} = \frac{1-\gamma^2}{1-\frac{\xi-\gamma}{m(\delta-\xi)}} \leq \frac{1}{1-\frac{\xi-\gamma}{m(\delta-\xi)}}
\end{align*}
It suffices to show that $\xi-\gamma \leq \delta-\xi$. This easily follows from the formulae in case 3). The proof for the case when $|\delta|\leq |\gamma|$ is obtained by replacing $\xi$ with $-\xi$.
\end{proof}

Like in Section \ref{sec: proofCuboidal}, the observer quantizes the measurements $y_t$ according to a $2^{nR}-$regular lattice quantizer with bin width $\delta$. In order for the controller to know $y_t$ to within a resolution of $\delta$, it is not hard to see that one needs $\delta 2^{nR} > 2\sqrt{P_t^{11}}+v$. We begin by assuming that the rate $R$ is large enough to provide the same resolution $\delta$ on $y_t$ at each time $t$. The actual rate required to accomplish this will be calculated determining an asymptotic upper bound on $P_t^{11}$. So, at time $t$, the controller knows that $y_t$ to within a resolution $\delta$ and hence $x_t^{(1)}$ to within a resolution of $\delta+V$. Suppose $\sqrt{P_{t}^{11}}\gamma_t \leq x_t^{(1)} \leq \sqrt{P_{t}^{11}}\delta_t$, where $\sqrt{P_{t}^{11}}(\delta_t - \gamma_t) \leq \delta+V$. Then using Lemma \ref{lem: minVol} and noting that $\gamma_t\leq \xi_t\leq\delta_t$, we have
\begin{align*}
 a_t = m_x(\xi_t-\gamma_t)(\delta_t - \xi_t) &\leq \frac{m_x}{4}(\delta_t-\gamma_t)^2\\
\implies P_t^{11}a_t \leq \frac{m_x}{4}(\delta+V)^2
\end{align*}
Using this in \eqref{eq: proofEllip2}, we get
\begin{align}
 \label{eq: proofEllip3}
P_{t|t}^{11} = a_tP_t^{11} \leq \frac{m_x}{4}(\delta+V)^2
\end{align}

Combining Lemma \ref{lem: boundbt} and \eqref{eq: proofEllip3}, we get
\begin{subequations}
 \begin{align}
  \sqrt{P_{t|t}^{11}} &\leq \frac{\sqrt{m_x}}{2}(\delta+V) \\
  \sqrt{P_{t|t}^{ii}} &\leq \sqrt{\frac{m_x}{m_x-1}}\sqrt{P_{t}^{ii}},\,\,\,i\neq 1
 \end{align}
\label{eq: proofEllipMU}
\end{subequations}

In the following Lemma, we will develop an upper bound on the diagonal elements of $P_t$ which will help us determine an upper bound on $P_{t}^{11}$.
\begin{lem}
 \label{lem: proofEllipUB}
Let $\Delta_{e,0}\in \Re^{m_x}$ be such that $\Delta_{e,0}^{(i)} = P_0^{ii}$ for $1\leq i\leq m_x$ and suppose its evolution is governed by
\begin{align*}
 \Delta_{e,t+1} &= (1+\epsilon')^{\frac{1}{2}}\Fbar \Delta_{e,t|t} + \frac{W}{2\sqrt{\epsilon'}}\mones \\
 \Delta_{e,t|t}^{(i)} &= \left\{\begin{array}{cc}
		     \delta+V & i = 1\\
		     \theta \Delta_{e,t}^{(i)} & i\neq 1
                  \end{array}\right.
\end{align*}
where $\theta = \sqrt{\frac{m_x}{m_x-1}}$. Then $\sqrt{P_{t}^{ii}} \leq \Delta_{e,t}^{(i)}$ and $\sqrt{P_{t|t}^{ii}} \leq \Delta_{e,t|t}^{(i)}$ for all $t$ and $1\leq i\leq m_x$.
\end{lem}
\begin{proof}
 The proof follows by combining the observations from \eqref{eq: proofEllipTU}, \eqref{eq: proofEllip0}, \eqref{eq: proofEllipMU}.
\end{proof}
Note that the recursion for $\Delta_{e,t}$ above is very similar to that for $\Delta_t$ in Section \ref{sec: proofCuboidal}. So, the steady state value of $\Delta_{e,t}^{(1)}$ can be determined by a calculation similar to that in Lemma \ref{lem: determiningDelta}. The desired threshold for $R$ is obtained by letting $\delta\rightarrow\infty$ for a fixed $\epsilon'$. Since $\epsilon'$ can be made arbitrarily small, we get the following bound on $R$
\begin{align*}
 R > \frac{1}{n}\log\left[\sqrt{m_x}\sum_{i=1}^{m_x}|a_i|\theta^{i-1}\right]
\end{align*}
Now, we need to determine the exponent needed to track \eqref{eq: sysmodel} with a bounded mean squared error. In the absence of any measurements, it is easy to see from \eqref{eq: ellipTU} that the growth of $P_t$ is determined by the spectral radius of $\sqrt{1+\epsilon'}F$. Since $\epsilon'$ can be made arbitrarily small, inorder to track \eqref{eq: sysmodel} with a bounded mean squared error, we need an anytime exponent $n\beta > 2\log\rho(F)$. This completes the proof.

%% file: AppendixLimitingCase.tex
\def\mIu{\mathcal{I}_u}
\subsection{The Limiting Case}
\label{subsec: limitingCase}

There are several bounds in the Mathematics literature on the roots of a polynomial in terms of the polynomial coefficients, a standard and near optimal bound being the Fujiwara's bound which we state below.
\begin{lem}[Fujiwara's Bound]
 \label{lem: fujiwara}
 Consider the monic polynomial with complex coefficients $f(z) = z^m + c_1z^{m-1}+\ldots+z_m$ and let $\rho(f)$ denote the largest root in magnitude. Then
\begin{align*}
 \rho(f) \leq K(f) = 2\max\left\{|c_1|,|c_2|^{\frac{1}{2}},\ldots,|c_{m-1}|^{\frac{1}{m-1}},\left|\frac{c_m}{2}\right|^{\frac{1}{m}}\right\}
\end{align*}
\end{lem}

We will detail the proof for the case of scalar measurements. The extension to the vector measurements will then suggest itself.
Let $F$ is any $m_x$-dimensional square matrix and $f(z)$ denotes its characteristic polynomial. Then the following bounds hold (for details see \cite{Sluis})
\begin{align}
\label{eq: sluis1}  \rho(F) \leq \rho(\Fbar) \leq \frac{\rho(F)}{\sqrt[m]{2}-1},\,\,K(f) \leq 2\rho(\Fbar)
\end{align} 

By the hypothesis of the Lemma, the eigen values of $F_n$ are of the form $\{\mu_i^n\}_{i=1}^{m_x}$. To emphasize the fact that $F$ depends on $n$, we write it as $F_n$ and $a_i$ as $a_{i,n}$. Recall that the characteristic polynomial of $F_n$ is given by $f_n(z) = z^{m_x} + a_{1,n}z^{m_x-1}+\ldots+a_{m_x,n}$. Let $\mIu\triangleq\{i\,\,\arrowvert\,\, |\mu_i| \geq 1\}$, then the following is easy to prove
\begin{align}
\label{eq: appendix1}
\lim_{n\rightarrow\infty}\frac{|a_{i,n}|}{\left|a_{|\mIu|,n}\right|} = 0,\,i\neq |\mIu|,\,\,\lim_{n\rightarrow\infty}\frac{1}{n}\log_2\left|a_{|\mIu|,n}\right| = \sum_{i\in\mIu}\log_2|\mu_i|
\end{align}
From \eqref{eq: appendix1}, it is obvious that $\lim_{n\rightarrow\infty}R_n = \sum_{i\in\mIu}\log_2|\mu_i|$. The asymptotics of $R_{e,n}$, $R_{v,n}$ and $R_{ev,n}$ can be similarly derived. Also, from \eqref{eq: sluis1}, it is clear that $\lim_{n\rightarrow\infty}\frac{1}{n}\log\rho\bra{\Fbar_n} = \lim_{n\rightarrow\infty}\frac{1}{n}\log\rho\bra{F_n}$. The asymptotics of $\beta_n$ and $\beta_{v,n}$ now follow immediately.

%% file: IEEETAC.bbl
\begin{thebibliography}{10}

\bibitem{Hinfinity}
B.~Hassibi, A.H. Sayed, and T.~Kailath,
\newblock {\em Indefinite-Quadratic Estimation and Control},
\newblock SIAM, 1999.

\bibitem{Murray}
R.M. Murray, K.J. Astrom, S.P. Boyd, R.W. Brockett, and G.~Stein,
\newblock ``Future directions in control in an information-rich world,''
\newblock {\em Control Systems, IEEE}, vol. 23, no. 2, apr 2003.

\bibitem{borkar}
VS~Borkar and SK~Mitter,
\newblock ``Lqg control with communication constraints,''
\newblock {\em Communications, computation, control, and signal processing: a
  tribute to Thomas Kailath}, pp. 365--373.

\bibitem{WongI}
Wing~Shing Wong and R.W. Brockett,
\newblock ``Systems with finite communication bandwidth constraints. i. state
  estimation problems,''
\newblock {\em Automatic Control, IEEE Transactions on}, vol. 42, no. 9, sep
  1997.

\bibitem{WongII}
Wing~Shing Wong and R.W. Brockett,
\newblock ``Systems with finite communication bandwidth constraints. ii.
  stabilization with limited information feedback,''
\newblock {\em Automatic Control, IEEE Transactions on}, vol. 44, no. 5, may
  1999.

\bibitem{Nillson}
J~Nillson,
\newblock {\em Real-time control systems with delays},
\newblock Ph.D. thesis, Lund Institute of Technology, 1998.

\bibitem{Walsh}
G.~C. Walsh and H.~Ye,
\newblock ``Scheduling of networked control systems,''
\newblock {\em IEEE Control Systems Magazine}, vol. 21, Feb 2001.

\bibitem{Nair}
GN~Nair and RJ~Evans,
\newblock ``Stabilizability of stochastic linear systems with finite feedback
  data rates,''
\newblock {\em SIAM Journal on Control and Optimization}, vol. 43, no. 2, pp.
  413--436, 2005.

\bibitem{Matveev}
Alexey~S. Matveev and Andrey~V. Savkin,
\newblock {\em Estimation and Control over Communication Networks (Control
  Engineering)},
\newblock Birkhauser, 2007.

\bibitem{Minero}
P.~Minero, M.~Franceschetti, S.~Dey, and G.N. Nair,
\newblock ``Data rate theorem for stabilization over time-varying feedback
  channels,''
\newblock {\em Automatic Control, IEEE Transactions on}, vol. 54, no. 2, 2009.

\bibitem{Martins}
N.C. Martins, M.A. Dahleh, and N.~Elia,
\newblock ``Feedback stabilization of uncertain systems in the presence of a
  direct link,''
\newblock {\em Automatic Control, IEEE Transactions on}, vol. 51, no. 3, march
  2006.

\bibitem{Sinopoli}
B.~Sinopoli, L.~Schenato, M.~Franceschetti, K.~Poolla, M.I. Jordan, and S.S.
  Sastry,
\newblock ``Kalman filtering with intermittent observations,''
\newblock {\em Automatic Control, IEEE Transactions on}, vol. 49, no. 9, sept.
  2004.

\bibitem{Sahai}
Anant Sahai and Sanjoy Mitter,
\newblock ``The necessity and sufficiency of anytime capacity for stabilization
  of a linear system over a noisy communication link - part i: Scalar
  systems,''
\newblock {\em Information Theory, IEEE Transactions on}, vol. 52, no. 8, pp.
  3369--3395, 2006.

\bibitem{Schulman}
LJ~Schulman,
\newblock ``Coding for interactive communication,''
\newblock {\em Information Theory, IEEE Transactions on}, vol. 42, no. 6, pp.
  1745 -- 1756, 1996.

\bibitem{Palaiyanur}
H~Palaiyanur and A~Sahai,
\newblock ``A simple encoding and decoding strategy for stabilization over
  discrete memoryless channels,''
\newblock {\em Proceedings of the Allerton Conference on Control, Communication
  and Computing}, 2005.

\bibitem{Ostrovsky}
R~Ostrovsky, Y~Rabani, and LJ~Schulman,
\newblock ``Error-correcting codes for automatic control,''
\newblock {\em Information Theory, IEEE Transactions on}, vol. 55, no. 7, pp.
  2931 -- 2941, 2009.

\bibitem{Yuksel}
Serdar Y\"uksel,
\newblock ``A random time stochastic drift result and application to stochastic
  stabilization over noisy channels,''
\newblock {\em Communication, Control, and Computing, 2009. Allerton 2009. 47th
  Annual Allerton Conference on}, pp. 628--635, 2009.

\bibitem{Simsek}
T.~Simsek, R.~Jain, and P.~Varaiya,
\newblock ``Scalar estimation and control with noisy binary observations,''
\newblock {\em Automatic Control, IEEE Transactions on}, vol. 49, no. 9, pp.
  1598 -- 1603, 2004.

\bibitem{Shannon}
CE~Shannon,
\newblock ``A mathematical theory of communication,''
\newblock {\em Bell System Technical Journal}, vol. 27, pp. 379 -- 423 and 623
  -- 656, July and Oct 1948.

\bibitem{Braverman}
Mark Braverman and Anup Rao,
\newblock ``Towards coding for maximum errors in interactive communication,''
\newblock in {\em Proceedings of the 43rd annual ACM symposium on Theory of
  computing}, 2011, STOC '11.

\bibitem{Gelles}
Ran Gelles and Amit Sahai,
\newblock ``Potent tree codes and their applications: Coding for interactive
  communication, revisited,''
\newblock {\em CoRR}, vol. abs/1104.0739, 2011.

\bibitem{Moitra}
Ankur Moitra,
\newblock ``Efficiently coding for interactive communication,''
\newblock {\em Electronic Colloquium on Computational Complexity (ECCC)}, 2011.

\bibitem{SerdarIEEETAC2011}
S.~Y\"uksel and T.~Ba\c{s}ar,
\newblock ``Control over noisy forward and reverse channels,''
\newblock {\em Automatic Control, IEEE Transactions on}, vol. 56, no. 5, may
  2011.

\bibitem{Gallager}
Robert~G. Gallager,
\newblock {\em Information Theory and Reliable Communication},
\newblock 1968.

\bibitem{Barg}
A.~Barg and Jr. Forney, G.D.,
\newblock ``Random codes: minimum distances and error exponents,''
\newblock {\em Information Theory, IEEE Transactions on}, vol. 48, no. 9, pp.
  2568 -- 2573, sep 2002.

\bibitem{Shamai}
Igal Sason and Shlomo Shamai,
\newblock ``Performance analysis of linear codes under maximum-likelihood
  decoding: A tutorial,''
\newblock {\em FNT in Communications and Information Theory}, vol. 3, no. 1/2,
  pp. 1--222, 2006.

\bibitem{SahaiVec}
Anant Sahai and Sanjoy~K. Mitter,
\newblock ``The necessity and sufficiency of anytime capacity for stabilization
  of a linear system over a noisy communication link, part ii: vector
  systems,''
\newblock {\em http://arxiv.org/abs/cs/0610146}, vol. abs/cs/0610146, 2006.

\bibitem{Shalom}
Y.~Bar-Shalom and E.~Tse,
\newblock ``Dual effect, certainty equivalence, and separation in stochastic
  control,''
\newblock {\em Automatic Control, IEEE Transactions on}, vol. 19, no. 5, oct
  1974.

\bibitem{Schweppe}
F.~Schweppe,
\newblock ``Recursive state estimation: Unknown but bounded errors and system
  inputs,''
\newblock {\em Automatic Control, IEEE Transactions on}, vol. 13, no. 1, Feb.
  1968.

\bibitem{Kailath}
T.~Kailath,
\newblock {\em {Linear Systems}},
\newblock Prentice-Hall, Inc., Englewood Cliffs, N.J., 1980.

\bibitem{SahaiWhy}
A.~Sahai,
\newblock ``Why do block length and delay behave differently if feedback is
  present?,''
\newblock {\em Information Theory, IEEE Transactions on}, vol. 54, no. 5, may
  2008.

\bibitem{Franklin}
Gene Franklin, J.~D. Powell, and Abbas Emami-Naeini,
\newblock {\em Feedback Control of Dynamic Systems},
\newblock Pearson Prentice Hall, 5th, edition, 2006.

\bibitem{Kakhaki}
A.~Makhdoumi Kakhaki, H.~Karkeh Abadi, Pedram Pad, H.~Saeedi, Kasra Alishahi,
  and Farokh Marvasti,
\newblock ``Capacity achieving random sparse linear codes,''
\newblock {\em CoRR}, vol. abs/1102.4099, 2011.

\bibitem{Guler}
Osman G{\"u}ler and Filiz G{\"u}rtuna,
\newblock ``The extremal volume ellipsoids of convex bodies, their symmetry
  properties, and their determination in some special cases,''
\newblock {\em arXiv}, vol. math.MG, Sep 2007.

\bibitem{Sluis}
A~Sluis,
\newblock ``Upperbounds for roots of polynomials,''
\newblock {\em Numerische Mathematik}, vol. 15, no. 3, pp. 250--262, 1970.

\end{thebibliography}
